\documentclass[twoside, twocolumn, journal, 10pt]{IEEEtran}
\IEEEoverridecommandlockouts
\usepackage{hyperref}
\usepackage{xcolor}
\usepackage{array}
\usepackage{tikz}
\usepackage{cite}
\usepackage[section]{placeins}
\usepackage{bm}
\usepackage{subcaption}
\usepackage{float}
\usepackage{amsmath}
\usepackage{graphicx}
\usepackage{epstopdf}
\usepackage{amssymb}
\usepackage{amssymb,amsmath}
\usepackage{amsgen,amsfonts,amsbsy,amsthm}
\usepackage{latexsym}
\usepackage{helvet}
\usepackage[utf8]{inputenc}
\allowdisplaybreaks
\makeatletter
\def\BState{\State\hskip-\ALG@thistlm}
\makeatother
\DeclareMathOperator*{\argmin}{arg\,min}

\newcounter{defcounter}
\newcommand{\totset}{\mathcal{H}}

\newcommand{\inprod}[2]{\left\langle #1,#2 \right\rangle}
\newcommand{\bvec}[1]{\boldsymbol{#1}}
\newcommand{\real}{\mathbb{R}}

\newcommand{\opnorm}[1]{\left\|#1\right\|}
\newcommand{\abs}[1]{\left|#1\right|}
\newcommand{\diag}[1]{\texttt{diag}\left(#1\right)}

\newcommand{\expect}[1]{\mathbb{E}\left[#1\right]}

\newtheorem{lem}{Lemma}[section]
\newtheorem{thm}{Theorem}[section]

\newtheorem{define}{Definition}[section]
\title{Deterministic and Randomized Diffusion based Iterative Generalized Hard Thresholding (DiFIGHT) for Distributed Sparse Signal Recovery}
\begin{document}
\setlength\abovedisplayskip{0pt}
\setlength\belowdisplayskip{0pt}
\author{Samrat Mukhopadhyay$^1$, {\sl Student Member, IEEE}, and Mrityunjoy
Chakraborty$^2$, {\sl Senior Member, IEEE}

% E.Mail :
% $^1$bijitbijit@gmail.com, $^2$vinaychakravarthi@gmail.com,
% $^3$mrityun@ece.iitkgp.ernet.in}

%\thanks{Manuscript received November, 2016.}

\thanks{Authors $^1$ and $^2$ are with the department of Electronics and Electrical Communication
Engineering, Indian Institute of Technology, Kharagpur, INDIA(email : $^1$samratphysics@gmail.com, $^2$mrityun@ece.iitkgp.ernet.in).
}}
%\vspace{-5mm}
%\author{\IEEEauthorblockN{Samrat~Mukhopadhyay\IEEEauthorrefmark{1}, Siddhartha Satpathi\IEEEauthorrefmark{2} and Mrityunjoy Chakraborty\IEEEauthorrefmark{3}} \\
%  \IEEEauthorblockA{Department of Electronics and Electrical Communication Engineering,\\ Indian Institute of Technology, Kharagpur, INDIA\\
%    Email: \IEEEauthorrefmark{1}samratphysics@gmail.com,\IEEEauthorrefmark{2}sidd.piku@gmail.com,
%    \IEEEauthorrefmark{3}mrityun@ece.iitkgp.ernet.in}}%\vspace{-5mm}}
%\IEEEoverridecommandlockouts    
    %
\maketitle
\begin{abstract}
In this paper we propose a distributed iterated hard thresholding algorithm termed DiFIGHT over a network that is built on the diffusion mechanism and also propose a modification of the proposed algorithm, termed MoDiFIGHT, that has low complexity in terms of communication in the network. We additionally propose four different strategies termed RP, RNP, RGP$_r$, and RGNP$_r$ that are used to randomly select a subset of nodes that are subsequently activated to take part in the distributed algorithm, so as to reduce the mean number of communications during the run of the distributed algorithm. We present theoretical estimates of the long run communication per unit time for these different strategies, when used by the two proposed algorithms. Also, we present analysis of the two proposed algorithms and provide provable bounds on their recovery performance with or without using the random node selection strategies. Finally we use numerical studies to show that both when the random strategies are used as well as when they are not used, the proposed algorithms display performances  far superior to distributed IHT algorithm using consensus mechanism .  
\end{abstract}
\section{Introduction}
%
% This is another area where we want to propose greedy algorithms for sparse recovery and see if it outperforms its noncooperative counterparts in some manner. 
%
In the distributed setting there is a network of nodes, where each node $v\in \{1,2,\cdots,\ L\}$, individually solves the following problem: \begin{align*}
\min_{\bvec{z}\in \real^N}\ f_i(\bvec{z})
\mathrm{s.t.}\ \opnorm{\bvec{z}}_0 \le K.
\end{align*}
Here the functions $f_v,\ v=1,2,\cdots,\ L$ are cost functions which are generally chosen to satisfy some kind of restricted convexity assumptions, i.e they are generally designed so that their curvatures have some specific properties. 
For example, in the distributed compressed sensing setting, a node $v$ measures a $K$-sparse vector $\bvec{x}\in \real^n$, and stores the $m$ dimensional ($m<n$) measurement as $\bvec{y}_v=\bvec{\Phi}_v\bvec{x}+\bvec{e}_v$, where $\bvec{e}_v$ is measurement noise. A suitable cost function in this case is $f_v(\bvec{z}) = \opnorm{\bvec{y}_v - \bvec{\Phi}_v\bvec{z}}_2^2$, and to impose conditions on its curvature, the matrix $\bvec{\Phi}_v$ is assumed to satisfy some kind of restricted isometry property~\cite{foucart2013mathematical}.
However, in the collaborative, or distributed setting, the nodes do not work alone and sparse recovery algorithms working at neighboring nodes exchange information among themselves during the run of the algorithm. This information exchange through collaboration helps the true estimate to emerge, often in a faster or in other more advantageous ways compared to non-cooperative setting.

 There are many practical problems which naturally fits in this distributed setting, such as the problem of distributed classification in machine learning~\cite{forero2010consensus}, or data fitting in statistics~\cite{chaudhuri2009privacy} where nodes may contain confidential data (medical records, transaction records etc.) and thus cannot transfer data to a centralized processor. In this scenario, thus it is assumed that their is no fusion center available, and all the operations have to be performed locally.The literature on distributed recovery is relatively recent~\cite{mateos2010distributed,sundman2012greedy,mota2012distributed}. These first contributions propose natural ways to distribute known centralized methods, and obtain interesting results in
terms of convergence and estimation performance. However they do not consider the
problem of the insufficient computation and memory resources. Distributed basis pursuit algorithms for
sparse approximations when the measurement matrices are not globally known have
been studied in~\cite{mota2012distributed,mateos2010distributed}. In these algorithms, sensors collaborate to estimate the original vector, and, at each iteration, they update this estimate based on communication with
their neighbors in the network. Most of these algorithms fall into the following families of algorithms: distributed subgradient methods (DSM)~\cite{nedic2010constrained,patterson2014distributed}, distributed alternating direction method of multipliers (ADMM)~\cite{boyd2011distributed}, and distributed iterative soft thresholding  algorithm (DISTA)~\cite{ravazzi2014energy,ravazzi2015distributed,lorenzo2016next}. All of these algorithms, in one form or other use the \emph{consensus} optimization paradigm, where the nodes in a neighborhood cooperatively minimize a global cost function while minimizing their individual local cost function. However, in the literature of adaptive networks, there is a different family of algorithms, called \emph{diffusion}, which are studied extensively by Sayed \emph{et. al}~\cite{tu2012diffusion,chen2012diffusion,di2013sparse,sayed2013diffusion,chen2015learning1,chen2015learning2}, and is shown to exhibit superior performances compared to consensus strategies~\cite{tu2012diffusion}, and also to outperform all noncooperative strategies. These strategies can be traced back to the generalized distributed communication and processing based model for distributed computation, proposed by Tsitsiklis~\cite{tsitsiklis1986distributed}. It is only recently that distributed sparse recovery algorithms have been designed to incorporate diffusion as the underlying mechanism. Patterson \emph{et al}~\cite{patterson2014distributed} have designed the distributed hard thresholding (DIHT) and the consensus based distributed hard thresholding (CB-DIHT) algorithms where one parent node forms a spanning tree and  over several time steps collects estimates of the gradients of the functions from all the nodes in the network. This strategy is a reminiscent of the diffusion mechanism. More recently, another distributed hard thresholding algorithm DiHaT is proposed and analyzed by Chouvardas \emph{et al}~\cite{chouvardas2015greedy}. Also Zaki \emph{et al}~\cite{zaki2017greedy} have proposed and analyzed a greedy distributed algorithm called the network gradient pursuit (NGP), and Zaki \emph{et al}~\cite{zaki2017estimate} have analyzed the distributed hard thresholding pursuit (DHTP) algorithm, originally proposed by Chouvardas~\cite{chouvardas2015greedy}. All of these algorithms have diffusion as the underlying mechanism.

On a different front, recently Ravazzi~\emph{et al}~\cite{ravazzi2016randomized} have proposed distributed algorithms with low communication overhead where only a few nodes are activated at each time step. However, the algorithms are modified from consensus IHT algorithm. It is the goal of this paper to propose and analyze a distributed IHT algorithm that minimizes general convex functions by using diffusion as its underlying mechanism and to modify it to generate algorithms where only a few nodes are selected per time step, resulting in reduced communication complexity. Specifically:\begin{itemize}
\item We propose and describe a distributed IHT algorithm termed DiFIGHT that uses diffusion mechanism to minimize general convex functions available to individual nodes.
\item We also propose and describe a simple low complexity modification of the DiFIGHT algorithm, termed MoDiFIGHT that, unlike diffusion, exchanges only estimates and thus uses less communication bandwidth.
\item We propose four strategies that are used to randomly select and activate only a subset of nodes at each time step, thus reducing communication overhead, and also give theoretical estimates on the long run communication overhead per unit time required by the two different algorithms when these strategies are used.
\item We analyze the algorithms with and without using the random node selection strategies and provide provable performance bounds.
\item We numerically evaluate the performances of these different algorithms with and without random node selection strategies and establish the superiority of diffusion mechanism over its consensus counterparts.
\end{itemize}
\section{Notation}
The following notations have been used throughout the paper :`$t$' in
superscript indicates transposition of matrices / vectors.%
%\item $\|\cdot\|_p$ ($p>0$) denotes the $l^p$ norm of a vector $\bvec{v}$, i.e., $\|\bvec{v}\|_p=\left(\sum_{i=1}^n |v_i|^p\right)^{1/p}$.
%$\bvec{\Phi}\in \mathbb{R}^{M\times N}$ denotes a matrix ($M<N$) and the $i$ th column of $\bvec{\Phi}$ is denoted
%by $\bm{\phi}_i,\ i=1,2,\cdots,\ n$. 
$\totset$ denotes the set of all the indices $\{1,2,\cdots,\ n\}$. For any $S\subseteq \totset$,
$\bvec{x}_S$ denotes the vector $\bvec{x}$ restricted to $S$,
i.e., $\bvec{x}_S$ consists of those entries of $\bvec{x}$ that have indices belonging to $S$. 
%Similarly, $\bvec{\Phi}_S$ denotes
%the submatrix of $\bvec{\Phi}$ formed with the columns of
%$\bvec{\Phi}$ restricted to the index set $S$. 
$\bvec{1}$ denotes a $L\times 1$ vector of $1$'s. The operator $H_K(\cdot)$ returns the \emph{$K$-best approximation} of a vector, i.e., for any vector $\bvec{x}\in \real^n$, $H_K(\bvec{x})=\argmin_{\bvec{z}\in \real^n:\opnorm{z}_0\le K}\opnorm{\bvec{z - x}}_2$. We denote by $\nabla_S f$, the restricted vector $(\nabla f)_S$, for any $S\subseteq \totset$. Also, we denote by $\nabla_K f$, the vector $\nabla f$ restricted to the subset corresponding to its $K$ magnitude-wise largest coordinates. The symmetric difference $\Delta$, between two sets $A, B$, is defined as $A\Delta B: = (A\setminus B)\cup (B\setminus B)$. 
\section{Diffusion based Hard Thresholding}
\label{sec:proposed-algorithm}
\subsection{Deterministic strategies}
\label{sec:deterministic-strategies}
The DiFfusion based Iterative Generalized Hard Thresholding (DiFIGHT) and its low complexity modification, the Modified DiFfusion based Iterative Generalized Hard Thresholding (MoDiFIGHT) algorithm are described in Table~\ref{tab:DiFIGHT-MoDiFIGHT}. There are $L$ nodes in the network. The combining matrix $\bvec{A}\in \real^{L\times L}$ specifies the connectivity between the different nodes in the network. The $(i,j)^{\mathrm{th}}$ entry of $\bvec{A}$, denoted by $a_{ij}\in [0,1]$, is the weight of the edge between nodes $i,j$. The graph $V$ formed by the nodes and the matrix $\bvec{A}$ is assumed to be undirected. Furthermore, the matrix $\bvec{A}$ is assumed to be left stochastic, i.e., $\bvec{A}^t\bvec{1}=\bvec{1}$. For any node $v$ in the graph represented by the combination matrix $\bvec{A}$, the \emph{neighborhood} of $v$ is denoted by $\mathcal{N}_v$, defined as $\mathcal
N_v=\{u\in V:a_{vu}>0\}$. The elements of $\mathcal{N}_v$ are called the neighbors of $v$. We assume that for each node $v$, $a_{vv}>0$, so that each node is a neighbor of itself. Each node $i$ has the function $f_i(\cdot)$ available with it. 

%In table we describe the MoDiFIGHT which is a low complexity modification of DiFIGHT as in the communications in this algorithm, the gradient update is first hard thresholded and then sent to respective neighbors.
%We assume that $\exists \bvec{x}^\star$, with $\opnorm{\bvec{x}^\star}_0 \le K$ with $\Lambda = \support(\bvec{x}^\star)$, such that $\nabla_\Lambda f_i(\bvec{x}^\star) = \bvec{0}_{K\times 1}$. This model is called single task model where each of the nodes in the network is trying to estimate the same unknown vector. The goal of the algorithm is to recover $\bvec{x}^\star$.
%
\begin{table}[t!]
%\begin{subfigure}{0.45\textwidth}
\caption{\textsc{Algorithm}: DiFIGHT and MoDiFIGHT}
\label{tab:DiFIGHT-MoDiFIGHT}
%\hspace{-1cm}
\begin{tabular}{p{9cm}}
%\hspace{5mm}
\centering
\hrulefill
\begin{description}
\item[\textbf{Input:}]\ Number of nodes $L$, the combining matrix $\bvec{A}$ such that $\bvec{A}^t\bvec{1} = \bvec{1}$, sparsity level $K$; Initial estimates $\bvec{x}_i^0,\ 1\le i\le L$; step sizes $\mu_i>0,\ i=1,2,\cdots,\ L$;  maximum number of iterations $k_{\mathrm{it}}$;
\end{description}
%\hrulefill \emph{\textsc{DiFIGHT}} \hrulefill
\begin{description}
\item[\textbf{While}]($k<k_{\mathrm{it}}$) 
%iterate
%%\begin{description}
%\begin{itemize}
%\item $n\in 3\mathbb{N},\ i=1,2,\cdots,\ L,$ 
%\begin{itemize}
%\item[]\ \   $\bvec{\zeta}_i^n=\sum_{j=1}^L a_{1,ji}\bvec{x}^n_j$
%\end{itemize}
%\end{itemize}
\begin{description}
\item[\textbf{For}] $i=1$ to $L$
%\item[]\    \   $\bvec{\zeta}_i^n=\sum_{j=1}^L a_{1,ji}\bvec{x}^n_j$
\item[]$\bvec{\psi}_i^{k+1}=\bvec{x}_i^k - \mu_i\nabla f_i(\bvec{x}_i^k)$
\item[]$\hat{\bvec{\psi}}_i^{k+1}=\left\{\begin{array}{ll}
\bvec{\psi}_i^{k+1}, & \mbox{DiFIGHT}\\
H_K\left(\bvec{\psi}_i^{k+1}\right), & \mbox{MoDiFIGHT}
\end{array}\right.$ 
\item[]\item[\  \textbf{End For}]
\item[\textbf{For}] $i=1$ to $L$
\item[]$\bvec{x}_i^{k+1}=H_K\left(\sum_{j=1}^L a_{ji}\hat{\bvec{\psi}}^{k+1}_j\right)$
\item[]\item[\  \textbf{End For}]
\item[]$k=k+1$
\end{description}
\item[\textbf{End While}]
\end{description}
\hrulefill
%\hrulefill \emph{\textsc{DIFHTP}} \hrulefill
%\begin{description}
%\item[\textbf{While}]($n<n_{\mathrm{it}}$) 
%%iterate
%%%\begin{description}
%%\begin{itemize}
%%\item $n\in 3\mathbb{N},\ i=1,2,\cdots,\ L,$ 
%%\begin{itemize}
%%\item[]\ \   $\bvec{\zeta}_i^n=\sum_{j=1}^L a_{1,ji}\bvec{x}^n_j$
%%\end{itemize}
%%\end{itemize}
%\begin{description}
%\item[\textbf{For}] $i=1$ to $L$
%%\item[]\    \   $\bvec{\zeta}_i^n=\sum_{j=1}^L a_{1,ji}\bvec{x}^n_j$
%\item[]$\bvec{\psi}_i^{n+1}=\bvec{x}_i^n - \mu_i\nabla f_i(\bvec{x}_i^n)$
%\item[]$\Lambda_i^{n+1}$ is the set of indices of the $K$ largest absolute values of $\sum_{j=1}^L a_{ji}\bvec{\psi}_i^{n+1}$
%\item[]$\bvec{x}_i^{n+1} = \argmin_{\bvec{z}\in \real^N: \support(\bvec{z})=\Lambda_i^{n+1}}f_i(\bvec{z})$
%\item[]\item[\  \textbf{End For}]
%\item[]$n=n+1$
%\end{description}
%\item[\textbf{End While}]
%\end{description}
%%\hrulefill
%%\begin{description}
%%\item Stop the algorithm if $n<n_{\mathrm{it}},\ \epsilon_n=0$, and put $\bvec{x}^k=\bvec{x}^n$, $\forall\ n>k$. 
%%\end{description}
%\hrulefill
%\begin{description}
%\item[\textbf{Output:}]\  \   $\bvec{x_i^n},\ i=1,2,\cdots,\ L$
%\end{description}
\hrulefill
\end{tabular}
\end{table}
\subsection{Randomized strategies}
\label{sec:randomized-strategies}
We also propose IHT based diffusion algorithms where all the nodes need not participate in the diffusion process at each time step. This absence of participation results into significant reduction in the amount of communication between the neighbors of the network, that would otherwise be required while exchanging values of estimates and gradient vectors. Inspired by Ravazzi \emph{et al}~\cite{ravazzi2015distributed}, we propose the following four different strategies for selecting the participating nodes:
\begin{enumerate}
\item\textbf{Random Persistence (RP):} In this strategy, at a time step $k$, only one node is selected at random according to a probability distribution $\{p_1\cdots,\ p_L\}$ over the nodes in the network. The probability distribution satisfies $p_v>0$ for each node $v$ in the network, and $\sum_{v\in \mathcal{V}}p_v = 1$, implying that each node has a positive probability of getting selected at a time step. Thus the selected group is $G=\{v\}$.
\item\textbf{Random Neighborhood Persistence (RNP):} As in the RP strategy, in this strategy too, at a time step $k$, a node $v$ is selected with probability $p_v$, where the probability distribution satisfies the same conditions as in the RP strategy. However, unlike the RP strategy, the neighborhood $\mathcal{N}_v$ of the selected node $v$ is also selected for participation in the diffusion process. Thus the selected group is $G = \{v\}\cup \mathcal{N}_v$.
\item\textbf{Random Group Persistence of order $r$ (RGP$_r$):} In this strategy a group $G$ of $r$ nodes is selected according to a probability distribution $\{p_{G}\}$ over all possible $\binom{L}{r}$ groups of nodes size $r$. Here the probability distribution is chosen such that $p_G>0$ for all such groups, and $\sum_{G\in \mathcal{G}_r}p_G = 1$, where $\mathcal{G}_r$ is the collection of all subsets of $\{1,2,\cdots,\ L\}$ of size $r$. Here the selected group of nodes is $G.$
\item\textbf{Random Group Neighborhood Persistence of order $r$ (RGNP$_r$):} In this strategy, a group of nodes $\widetilde{G}$ is chosen with probability $p_{\widetilde{G}}$ and $\widetilde{G}$ as well the union of their neighborhoods is selected. The probability distribution $p_{\widetilde{G}}$ is assumed to satisfy the same conditions as in the RGP$_r$ strategy. The selected group is $G = \widetilde{G}\cup_{u\in \widetilde{G}}\mathcal{N}_u.$
\end{enumerate}
 Once a group is selected, the diffusion process is applied to all the nodes in the group. The resulting algorithms are described in Table~\ref{tab:randomized-DiFIGHT-MoDiFIGHT}. 
 %Note that a node at a time is sampled from a discrete probability distribution given by the probability mass function $\bvec{\pi}_{net}=\left(p_1\ p_2\cdots\ p_L\right)^t$ where the probability of a node $v$ being selected at any time step $n$ is given by $p_{v}$. The probability distribution is chosen in a way such that $p_v>0\ \forall v=1,2,\cdots,\ L$. A network of nodes that satisfies this condition is called a randomly persistent network~\cite{ravazzi2015distributed}.
%
\begin{table}[ht!]
\caption{\textsc{Algorithm}: Randomized DiFIGHT and MoDiFIGHT }
\label{tab:randomized-DiFIGHT-MoDiFIGHT}
%\hspace{1cm}
\centering
\begin{tabular}{p{8cm}}
%\hspace{5mm}
\centering
\hrulefill
\begin{description}
\item[\textbf{Input:}]\ Number of nodes $L$, the combining matrix $\bvec{A}$ such that $\bvec{A}^t\bvec{1} = \bvec{1}$, sparsity level $K$; Initial estimates $\bvec{x}_i^0,\ 1\le i\le L$; step sizes $\mu_i>0,\ i=1,2,\cdots,\ L$;  maximum number of iterations $k_{\mathrm{it}}$;
\end{description}
\hrulefill
%\hrulefill \emph{\textsc{Randomized DiFIGHT}} \hrulefill
\begin{description}
\item[\textbf{While}]($k<k_{\mathrm{it}}$) 
%iterate
%%\begin{description}
%\begin{itemize}
%\item $n\in 3\mathbb{N},\ i=1,2,\cdots,\ L,$ 
%\begin{itemize}
%\item[]\ \   $\bvec{\zeta}_i^n=\sum_{j=1}^L a_{1,ji}\bvec{x}^n_j$
%\end{itemize}
%\end{itemize}
%\begin{description}
%\item RP: $G = v$, where $v$ is sampled from distribution $\bvec{\pi}_{net}$
%\item RNP: $G = \{v\} \cup \mathcal{N}_v$, where $v$ is sampled from distribution $\bvec{\pi}_{net}$
%\item RGP: $G = \bvec{v}_C$, where $C$ is a group of nodes sampled from distribution $\bvec{\pi}_{\mathrm{group}}$
%\item RNGP: $G = \bvec{v}_C\cup_{u\in \bvec{v}_C} \mathcal{N}_u$ where $C$ is a group of nodes sampled from distribution $\bvec{\pi}_{\mathrm{group}}$
\begin{description}
\item[\textbf{For}]$v = 1 : L$
\ \ \ \item[\ \textbf{if}]$v\in G\ \mathrm{or}\ \mathcal{N}_v\cap G\ne \emptyset$
%\item[]\    \   $\bvec{\zeta}_i^n=\sum_{j=1}^L a_{1,ji}\bvec{x}^n_j$
\ \ \ \ \ \item[]$\bvec{\psi}_v^{k+1}=\bvec{x}_v^k - \mu_v\nabla f_v(\bvec{x}_v^k)$
\ \ \ \ \ \item[]$\hat{\bvec{\psi}}_v^{k+1}=\left\{\begin{array}{ll}
\bvec{\psi}_v^{k+1}, & \mbox{DiFIGHT}\\
H_K\left(\bvec{\psi}_v^{k+1}\right), & \mbox{MoDiFIGHT}
\end{array}
\right.$
\ \ \ \item[]\item[\ \textbf{end if}]
\ \item[]\item[\  \textbf{End For}]
\item[\textbf{For}]$v \in G$
\ \item[]$\bvec{x}_v^{k+1}=H_K\left(\sum_{u\in \mathcal{N}_v} a_{uv}\hat{\bvec{\psi}}^{k+1}_u\right)$
\item[]\item[\  \textbf{End For}]
\end{description}
\item[]$\bvec{x}_u^{k+1} = \bvec{x}^k_u\ \forall u\notin G$
\item[]$k=k+1$
%\end{description}
\item[\textbf{End While}]
\end{description}
\hrulefill
\end{tabular}
\end{table}
%
%
%\begin{table*}

%\end{table*}
\subsection{Discussion on communication complexities}
\label{sec:communication-complexity-discussion}
\begin{table*}
	\caption{Communication complexities for uniform distribution}
	\label{tab:communication-complexities}
	\begin{center}
		\begin{tabular}{|c|>{\centering}m{4cm}|c|c|}
%			\hline
%			\ & \multicolumn{3}{|c|}{Uniform distribution}\\
			\hline
			Strategy & $\pi_v$ & $\frac{T(v)}{C}$ & $\frac{R(v)}{C}$ \\
			\hline
			RP & $\displaystyle\frac{1}{L}$ & $\displaystyle\frac{d_v}{L}$ & $\displaystyle\frac{d_v}{L}$\\
			\hline
			RNP & $\displaystyle\frac{1+d_v}{L}$ & $\displaystyle \frac{d_v}{L} + \sum_{u\in \mathcal{N}_v}\frac{d_u}{L}$ & $\displaystyle\frac{d_v(1+d_v)}{L}$\\
			\hline
			RGP$_r$ & $\displaystyle\frac{r}{L}$ & $\displaystyle \frac{d_v r}{L}$ & $\displaystyle \frac{d_v r}{L}$\\
			\hline
			RGNP$_r$ & $\displaystyle 1-\prod_{k=0}^{r-1}\left(1-\frac{d_v+1}{L-k}\right)$ if $L-(d_v+1)\ge r$, else, $1$ & $\displaystyle \sum_{u \in \mathcal{N}_v}\pi_u$ & $\displaystyle d_v\pi_v$\\
			\hline
		\end{tabular}
	\end{center}
\end{table*}
We present a comparative discussion on the mean number of communication required by the nodes for the different algorithms. Note that the communication complexity depends on both the diffusion mechanism as well as the strategy of selection of group of participating nodes. Since the total communication complexity of the network is just the sum of the complexities of individual nodes, we focus on finding out the communication complexity for some fixed node $v$. We consider the time average of the number of messages transmitted and sent by the node, which are denotes by $T(v) = \lim_{t\to \infty}\frac{\sum_{k=1}^t T_k(v)}{t}$ and $R_v = \lim_{t\to \infty}\frac{\sum_{k=1}^t R_k(v)}{t}$ respectively, where $T_k(v), R_k(v)$ are the number of messages transmitted and received respectively, by the node $v$ at time step $k$. Note that, in all the four strategies adopted, for each node $v$, $\{T_k(v)\}_{k\ge 0}\left(\{R_k(v)\}_{k\ge 0}\right)$ is an independent and identically distributed (i.i.d.) sequence of bounded random variables, which ensures, by the strong law of large numbers (SLLN), that the limits $T(v)$ and $R(v)$ exist for all nodes $v$. To carry out the analysis, we denote by $d_v$ the degree of the node $v$, which is the number of neighbors of node $v$, excluding itself.

Before proceeding to find expressions for $T(v),R(v)$ for the different algorithms, we point out that the communication complexities of DiFIGHT and MoDiFIGHT are intrinsically distinct because of the fact that for an update to occur, in DiFIGHT, each transmitting node transmits $n$ values, whereas, in MoDiFIGHT each transmitting node has to perform $2K$ transmissions, $K$ for the support indices and $K$ for the values corresponding to those indices. Therefore, if $K<<n/2$, the number of communications in MoDiFIGHT can be much smaller than that of DiFIGHT.  
%We first discuss the intrinsic difference between the communication complexities of the different algorithms for any strategy. The difference lies in the number of transmitted and sent values of updates per exchange. For the DiFIGHT algorithm, each node obtains the estimates as well as the gradients computed locally by their neighbors. This requires  a communication of $2K+N$ values, $K$ for the indices of the support of the estimate, $K$ for the nonzero values of the estimate, and $N$ for the gradient. In MoDiFIGHT this number of communications reduces significantly as only the estimates are communicated, which requires only $2K$ values to be communicated between any pair of nodes.   

We first analyze the communication complexities of the deterministic diffusion algorithms. In this case, all the nodes of the network are chosen at every time step, so that $T_k(v)=T(v),\ R_k(v)=R(v), \forall k\ge 0$. Clearly, in DiFIGHT $T(v)=R(v)=nd_v$, while in MoDiFIGHT, $T(v)=R(v)=2Kd_v$.  

We now carry out the analysis of $T(v),\ R(v)$ and therefore for $T_k(v),\ R_k(v)$ for the randomized algorithms at a time step $k$. To do that, we denote $C=n,2K$ for deterministic DiFIGHT and MoDiFIGT, respectively. 
%Also, at a time step $k$, a node $v$ is called a \emph{transmitting node} it can transmit at time step $k$ and is called a \emph{participating node} if it 
 
First let us consider the calculation of $R_k(v)$. Observe that $R_k(v)=Cd_vI_k(v),$ where $I_k(v)$ is an indicator random variable taking value $1$ if $v\in G_k$ where $G_k$ is the group of nodes selected at time $k$ by the randomized algorithm (in which case $v$ is referred to as a \emph{participating node} at time $k$) and is $0$ otherwise. Clearly, for a fixed $v$, the sequence $\{I_k(v)\}_{k\ge0}$ is a sequence of i.i.d. random variables. Therefore, by SLLN, \begin{align}
\label{eq:Rv-expression}
 \frac{R(v)}{Cd_v} & = \lim_{t\to \infty}\frac{\sum_{k=0}^{t-1}I_k(v)}{t}=\expect{I_0(v)}=\pi_v\ a.s.,
\end{align}
where $\pi_v$ is the probability that the node $v$ is participating (at any time $k\ge 0$) and is called the \emph{participation probability}.

To calculate $T_k(v)$, observe that $T_k(v)$ is equal to $C$ times the number of nodes in the neighborhood of $v$ (distinct from $v$) participating at time $k$. Therefore, \begin{align}
T_k(v) & = C\sum_{u\in \mathcal{N}_v\setminus \{v\}}I_{k}(u)\nonumber\\
\implies \frac{T(v)}{C} & = \lim_{t\to \infty}\frac{\sum_{k=0}^{t-1}T_k(v)}{t}=\sum_{u\in \mathcal{N}_v\setminus \{v\}}\pi_u\ a.s.
\end{align} 

We now evaluate $\pi_v$ for the different randomized strategies proposed. We assume in the following that the probability of selection of a node $v$ is $p_v$ and the probability of selection of a group of nodes $G$ is $p_G$. Clearly, for uniformly random selections $p_v = 1/L,\ p_G=1/\binom{L}{\abs{G}}$. \begin{enumerate}
\item For the RP strategy, only one node can be selected at a time. Therefore, $\pi_v=p_v$. For uniformly random selection, $\pi_v=1/L$.
\item For the RNP strategy the node $v$ participates if either it is selected (w.p. $p_v$) or one of its neighbor is selected (w.p. $\sum_{u\in \mathcal{N}_n\setminus v}p_u$). Hence $\pi_v = \sum_{u\in \mathcal{N}_v}p_u$. For uniformly random selection $\pi_v = \frac{d_v + 1}{L}.$
\item For the RGP$_r$ strategy, the node $v$ participates if it belongs to a group of $r$ nodes $G$ which contains $v$. Since only one such group is selected, we have $\pi_v=\sum_{G:\abs{G}=r,\ v\in G}p_G$. For uniformly random selection $\pi_v=\frac{\binom{L-1}{r-1}}{\binom{L}{r}}=\frac{r}{L}$.
%Since there are $\binom{L-1}{r-1}$ different ways to choose a group of size $r$ that contains $v$, and since only one group is selected, we have $\pi_v = \frac{\binom{L-1}{r-1}}{\binom{L}{r}}=\frac{r}{L}$.
\item For the RGNP$_r$ strategy, the node $v$ participates if a group of $G$ of size $r$ is selected such that $v$ belongs to the neighborhood of the nodes in $G$. Therefore, $v$ participates if \begin{align}
v\in \cup_{u\in G}\mathcal{N}_u &\Leftrightarrow G\cap \mathcal{N}_v\ne \emptyset\nonumber\\
\implies\pi_v & =1-\sum_{\substack{G:\abs{G}=r,\\ G\cap \mathcal{N}_v=\emptyset}}p_G
\end{align}
For uniformly random selection, we have $\pi_v = 1 - \frac{\binom{L-(d_v+1)}{r}}{\binom{L}{r}}$. Note that calculation of this probability assumes that $L-(d_v+1)\ge r$. Otherwise, the node $v$ is present in every neighborhood and thus always participates, i.e. $\pi_v = 1$, which is an example of a highly connected node. 
\end{enumerate}
We enlist the values of $\pi_v,\ R(v)/C,\ T(v)/C$ for the different randomized strategies for uniform distribution in Table~\ref{tab:communication-complexities}.
\section{Theoretical result}
\label{sec;theoretical-results}
Let $\bvec{x}^\star$ be a $K$-sparse vector. In this section, we analyze how the distance of the iterates produced by the diffusion algorithms from the vector $\bvec{x}^\star$ changes with each iteration. For the purpose of our analysis a few assumptions are needed to be imposed on the functions $f_i,\ 1\le i\le L$.
\subsection{Preliminaries and assumptions}
\label{sec:preliminaries-assumptions}
\begin{define}[Restricted  Positive Definite Hessian]
Suppose that $f$ is a twice continuously differentiable function with Hessian $\nabla^2 f(\cdot)$. Then f is said to have a Restricted Positive Definite Hessian (RPDH) with constants $(\alpha_s,\ \beta_s)$, or $(\alpha_s,\ \beta_s)$-RPDH if the following holds:\begin{align}
\label{def:restricted-positive-definite-hessian}
\alpha_s\opnorm{\bvec{x}}_2^2 \le \bvec{x}^t\nabla^2f(\bvec{u}) \bvec{x} \le \beta_s\opnorm{\bvec{x}}_2^2
\end{align}
for all vectors $\bvec{x},\ \bvec{u}\in \real^n$ such that $\opnorm{\bvec{u}}_0\le s,\ \opnorm{\bvec{x}}_0\le s$.
\end{define}
This property is just a variant of the Stable Restricted Hessian (SRH) property defined in~\cite{bahmani2013greedy}, which bounds the curvature of $f$, when restricted to the union of all subspaces of sparse vectors of a given sparsity. To see the implication of the RPDH property, observe that the Hessian $\nabla^2f(\bvec{u})$ is a positive semidefinite matrix $\forall \bvec{u}$, so that it admits the unique eigen-decomposition $\bvec{Q}(\bvec{u})^t\bvec{D}(\bvec{u})\bvec{Q}(\bvec{u})$, where $\bvec{Q}(\bvec{u})$ is an orthogonal matrix and $\bvec{D}(\bvec{u})$ is a diagonal matrix. Writing $\bvec{\Phi}(\bvec{u}) = \bvec{D}^{1/2}(\bvec{u})\bvec{Q}(\bvec{u})$, we then see that $\nabla^2 f(\bvec{u}) = \bvec{\Phi}(\bvec{u})^t\bvec{\Phi}(\bvec{u})$. Then observe that the $(\alpha_s,\ \beta_s)-$RPDH property just implies that $\forall \bvec{x},\ \bvec{u}\in \real^n$ such that $\opnorm{\bvec{x}}_0\le s,\ \opnorm{\bvec{u}}_0\le s$, the matrix $\bvec{\Phi}(\bvec{u})$ satisfies: \begin{align*}
\alpha_s\opnorm{\bvec{x}}_2^2\le \opnorm{\bvec{\Phi}(\bvec{u})\bvec{x}}_2^2\le \beta_s\opnorm{\bvec{x}}_2^2.
\end{align*}  Thus RPDH is just a generalization of the well known Restricted Isometry Property (RIP)~\cite{candes_decoding_2005} to nonlinear operators. This RIP implication of the RPDH property is useful in proving the following lemma:
\begin{lem}
\label{lem:RPDH-inprod-inequality}
Let $\bvec{x},\ \bvec{y},\ \bvec{z}$ are vectors in $\real^n$ with supports $T_1,\ T_2,\ T_3$ respectively, and let $T = T_1\cup T_2\cup T_3$. Let $\rho$ be an arbitrary positive number. Also, let $\bvec{g}(\bvec{y},\bvec{z}):=\bvec{y} - \bvec{z} - \rho(\nabla f(\bvec{y}) - \nabla f(\bvec{z}))$. Then, \begin{enumerate}
\item \begin{align}
\label{eq:rpdh-inequality-result1}
\inprod{\bvec{x}}{\bvec{g}(\bvec{y},\bvec{z})} & \le \rho'_{\abs{T}}\opnorm{\bvec{x}}_2\opnorm{\bvec{y} - \bvec{z}}_2,
\end{align} and,
\item \begin{align}
\label{eq:rpdh-inequality-result2}
\opnorm{\left(g(\bvec{y},\bvec{z})\right)_{T_1}}_{2}\le \rho'_{\abs{T}}\opnorm{\bvec{y} - \bvec{z}}_2,
\end{align}
\end{enumerate}
where $\rho'_{\abs{T}}=\left(\abs{1 - \rho\delta^{(1)}_{|T|}} + \rho \delta^{(2)}_{\abs{T}}\right)$, $f$ satisfies the RPDH-$(\alpha_{\abs{T}},\ \beta_{\abs{T}})$ property, and $\delta^{(1)}_{\abs{T}} = \frac{\beta_{\abs{T}}+\alpha_{\abs{T}}}{2},\ \delta_{\abs{T}}^{(2)} = \frac{\beta_{\abs{T}}- \alpha_{\abs{T}}}{2}$.
\end{lem}
\begin{proof}
The key observation for the proof is the following: \begin{align*}
\nabla f(\bvec{y}) - \nabla f(\bvec{z}) & = \int_0^1 \nabla^2 f(\bvec{u})(\bvec{z}-\bvec{y}) d\tau\\
\ & = \int_0^1\bvec{\Phi}(\bvec{u})^t\bvec{\Phi}(\bvec{u})(\bvec{z}-\bvec{y}) d\tau,
\end{align*}
where, $\bvec{u} = \bvec{y} + \tau (\bvec{z} - \bvec{y})$ and $\bvec{\Phi}(\bvec{u})$ arises from the eigen-decomposition of $\nabla^2 f(\bvec{u})$ as discussed before. To prove 1) using the eigendecomposition of $\nabla^2f(\bvec{u})$, we write the inner product in question as $\int_0^1\inprod{\bvec{x}}{(\bvec{I} - \rho\bvec{\Phi}(\bvec{u})^t\bvec{\Phi}(\bvec{u}))(\bvec{y} - \bvec{z})} d\tau = \int_0^1\inprod{\bvec{x}_T}{(\bvec{I}_T - \rho\bvec{\Phi}_T(\bvec{u})^t\bvec{\Phi}_T(\bvec{u}))(\bvec{y}_T - \bvec{z}_T)} d\tau$. Then using the RPDH-$(\alpha_{|T|}, \beta_{|T|})$ property of $f$, or equivalently, the RIP like property for $\bvec{\Phi}(\bvec{u})$ of order $\abs{T}$, one finds, using Cauchy-Scwartz inequality, \begin{align}
\lefteqn{\inprod{\bvec{x}_T}{(\bvec{I}_T - \rho\bvec{\Phi}_T(\bvec{u})^t\bvec{\Phi}_T(\bvec{u}))(\bvec{y}_T - \bvec{z}_T)}} & & \nonumber\\
\ & \le \opnorm{\bvec{x}_T}_2\opnorm{\bvec{I}_T-\rho\bvec{\Phi}_T(\bvec{u})^t\bvec{\Phi}_T(\bvec{u})}_{2\to 2}\opnorm{\bvec{y}_T-\bvec{z}_T}_2 \nonumber\\
\label{eq:intermediate-inequality-lemma-rpdh-product}
\ & \stackrel{\psi} {\le}\left(\abs{1 - \rho\delta^{(1)}_{|T|}} + \rho \delta^{(2)}_{\abs{T}}\right)\opnorm{\bvec{x}}_2\opnorm{\bvec{y}-\bvec{z}}_2.
\end{align} Here step $\psi$ follows from the following observation: \begin{align*}
\lefteqn{\lambda_{\max}\left(\bvec{I}_T-\bvec{\Phi}_T(\bvec{u})^t\bvec{\Phi}_T(\bvec{u})\right)} & &\\
\ & \le \max\{\abs{1-\rho\alpha_{\abs{T}}},\ \abs{1-\rho\beta_{\abs{T}}}\}\\
\ & = \left\{
\begin{array}{rl}
1 - \rho\alpha_{\abs{T}} & \mbox{if}\ 0<\rho\le \frac{2}{\alpha_{\abs{T}}+\beta_{\abs{T}}}\\
\rho\beta_{\abs{T}} - 1 & \mbox{if}\ \rho> \frac{2}{\alpha_{\abs{T}}+\beta_{\abs{T}}}
\end{array}\right.
\end{align*} Since the RHS of the inequality~\eqref{eq:intermediate-inequality-lemma-rpdh-product} is independent of $\tau$, the final inequality~\eqref{eq:rpdh-inequality-result1} follows immediately.

For the proof of inequality~\eqref{eq:rpdh-inequality-result2}, first construct the vector $\bvec{u}\in \real^n$ such that $\bvec{u}_{T_1}=\bvec{g}(\bvec{y},\bvec{z})_{T_1}$, and $\bvec{u}_{T_1^C}=\bvec{0}_{T_1^C}$. Then, using the inequality~\eqref{eq:rpdh-inequality-result1}, one obtains \begin{align*}
\inprod{\bvec{u}}{\bvec{g}(\bvec{y},\bvec{z})} & \le \rho'_{\abs{T}}\opnorm{\bvec{u}}_2\opnorm{\bvec{y} - \bvec{z}}_2\\
\implies \opnorm{\bvec{u}_{T_1}}_2^2 & \le \rho'_{\abs{T}}\opnorm{\bvec{u}_{T_1}}_2\opnorm{\bvec{y} - \bvec{z}}_2,
\end{align*}
which, after cancellation of $\opnorm{\bvec{u}_{T_1}}_2$ from both sides of the above inequality results in the inequality~\eqref{eq:rpdh-inequality-result2}.
\end{proof}

Before proceeding to analyze the error sequence $\opnorm{\bvec{x}^{k+1}-\bvec{x}^\star}_2$, we recall a few definitions from the theory of non-negative matrices~\cite{seneta2006non}.
\begin{define}[Non-negative matrix]
A square matrix $\bvec{X}$ is said to be non-negative if for every pair of indices $i,j$, $(\bvec{X})_{ij}\ge 0$.
\end{define}

\begin{define}[Irreducible matrix]
\label{def:irreducible-matrix}
A square non-negative matrix $\bvec{X}$ is said to be irreducible, if for any pair of indices $i,j$, $\exists$ a positive integer $t_{ij}$ such that $(\bvec{X}^{t_{ij}})_{ij}>0$. 
\end{define}

We also recall the following classical result from Perron-Frobenius theory~\cite{seneta2006non}, which is going to be useful in our analysis.

\begin{thm}[Perron-Frobenius~\cite{seneta2006non} ]
\label{thm:perron-frobenius}
Let $\bvec{X}\in \real^{L\times L}$ be a non-negative irreducible matrix. Then, the following results hold:\begin{enumerate}
\item $\exists r>0$, such that $r$ is an eigenvalue of $\bvec{X}$, and $\abs{\lambda}\le r$, for any other eigenvalue $\lambda$ of $\bvec{X}$.
\item $r\in [\min_{i}\sum_j(\bvec{X})_{ij},\ \max_{i}\sum_j(\bvec{X})_{ij}]$.
\item $r$ has algebraic multiplicity $1$, and has strictly positive right and left eigenvectors $\bvec{u}, \bvec{w}^t$ respectively.
\item If $r, \lambda_2,\lambda_3,\ \cdots,\ \lambda_s$ are the distinct eigenvalues of $\bvec{X}$ with multiplicities $1,\ m_2,\cdots,\ m_s$, with $r>\abs{\lambda_2}>\cdots>\abs{\lambda_s}$, then, \begin{itemize}
\item If $\lambda_2\ne 0$, as $k\to \infty$, $\bvec{X}^k = r^k\bvec{u}\bvec{w}^t + o(k^{m_2-1}\abs{\lambda_2}^k)$.
\item If $\lambda_2 = 0$, $\forall k\ge L - 1$, $\bvec{X}^k = r^k\bvec{u}\bvec{w}^t$.
\end{itemize}
\end{enumerate}  
\end{thm}
We will also use the following simple but useful lemma:
\begin{lem}
\label{lem:irreducibility}
Let $\mathcal{G}$ be an undirected connected graph, with an associated non-negative weight matrix $\bvec{X}$. Then,\begin{enumerate}
\item $\bvec{X}^t$ is irreducible.
\item $\bvec{D}_1\bvec{X}\bvec{D}_2$ is irreducible for any two diagonal matrices $\bvec{D}_1,\ \bvec{D}_2$ which have strictly positive diagonal entries.
\item $\bvec{X} + \bvec{M}$ is irreducible for any non-negative matrix $\bvec{M}$.
\end{enumerate} 
\end{lem}
\begin{proof}
A short proof is delivered in Appendix~\ref{sec:appendix-proof-irreducibility-lemma}.
\end{proof}
We will further use the following lemma that will be useful to find upper bounds on the norm of the error between the iterates produced by an algorithm, and the target vector.
\begin{lem}
\label{lem:convergence-linear-inequality}
Let $\bvec{B}\in \real^{L\times L},\ \bvec{b}\in \real^L$ be a non-negative matrix and a non-negative vector, respectively. Let $\{\bvec{u}^k\}_{k\ge 0}$ be a sequence of non-negative vectors in $\real^L$ such that \begin{align*}
\bvec{u}^{k+1}\preccurlyeq \bvec{B}\bvec{u}^k + \bvec{b},\ k\ge 0.
\end{align*}
Then, if the matrix  $\bvec{B}$ is stable, and if $\bvec{u}$ be any limit point of the sequence $\{\bvec{u}^k\}_{n\ge 0}$, then, \begin{align*}
\bvec{u}\preccurlyeq (\bvec{I} - \bvec{B})^{-1}\bvec{ b}.
\end{align*}
\end{lem}
\begin{proof}
The proof is supplied in Appendix~\ref{sec:appendix-proof-lemma-convergence-linear-inequality }.
\end{proof}
Finally, we will use the following classical result in convergence of random sequences to show almost sure convergence of the randomized algorithms under certain conditions:
\begin{lem}[Almost sure convergence\cite{grimmett2001probability} ]
\label{lem:almost-sure-convergence-grimmett}
If $P_n(\epsilon) = \mathbb{P}(\abs{X_n - X}>\epsilon)$ satisfies $\sum_{n=1}^\infty P_n(\epsilon)<\infty$ for all $\epsilon>0$, then $X_n\stackrel{a.s.}{\to} X$.
\end{lem}
\subsection{Notation used in the main results}
\label{sec:notation-main-results}
We now proceed to analyze the evolution of the distance between $\bvec{x}^\star$ and the iterates produced by DiFIGHT. Before presenting the main results, we list the notation used hereafter in the paper:
\begin{itemize}
\item $\bvec{h}^k = \left[\opnorm{\bvec{x}^k_1 - \bvec{x}^\star}_2\ ,\cdots,\ \opnorm{\bvec{x}^k_L - \bvec{x}^\star}_2\right]^t,\ \forall k\ge 0$.
%\item $\bvec{\rho}_i = \abs{1-\mu_i\frac{\beta_{i,2K}+\alpha_{i,2K}}{2}} + \mu_i\frac{\beta_{i,2K} - \alpha_{i,2K}}{2},\ 1\le i\le L$.
\item At any step $k\ge 0,\ 1\le i\le L$, $\Lambda_i^{k}$ is the support set of $\bvec{x}_i^k$.
\item $\omega_i = \abs{1-\mu_i\frac{\beta_{i,3K}+\alpha_{i,3K}}{2}} + \mu_i\frac{\beta_{i,3K} - \alpha_{i,3K}}{2},\ 1\le i\le L$.
%\item $\bvec{D}_1 = \diag{\sqrt{\frac{2}{1-\rho_1^2}},\cdots,\ \sqrt{\frac{2}{1-\rho_L^2}}}$,
%\item $\bvec{D}_2 =  \diag{\frac{1}{1-\rho_1},\cdots,\ \frac{1}{1-\rho_L}}$,
\item $\bvec{\Omega} = \diag{\omega_1,\cdots,\ \omega_L},\ \bvec{M} = \diag{\mu_1,\cdots,\ \mu_L}.$
\item $\bvec{b} = [\opnorm{\nabla_{2K} f_1(\bvec{x}^\star)}_2,\cdots,\ \opnorm{\nabla_{2K} f_L(\bvec{x}^\star)}_2]^t$. 
%\item $\bvec{b}_2 = [\opnorm{\nabla_{K} f_1(\bvec{x}^\star)}_2,\cdots,\ \opnorm{\nabla_{K} f_L(\bvec{x}^\star)}_2]^t$
\item For any two vectors $\bvec{a}, \bvec{b}\in \real^n$, the inequality $\bvec{a} \preccurlyeq \bvec{b}$ implies that $a_i\le b_i,\ \forall i=1,2,\cdots,\ n$. 
\end{itemize}
\subsection{Main results}
\label{sec:main-results}
\subsubsection{Deterministic algorithms} 
The main results for the deterministic DiFIGHT algorithm are stated in theorem~\ref{thm:convergence-DiFIGHT}:
\begin{thm}
\label{thm:convergence-DiFIGHT}
Under the RPDH assumption, at any iteration $k$, the iterate produced by DiFIGHT as well as MoDIFIGHT satisfies the following inequlaity:\begin{align}
\label{eq:DiFIGHT-MoDIFIGHT-vector-inequality}
\bvec{h}^{k+1}\preccurlyeq \alpha\bvec{A}^t\bvec{\Omega} \bvec{h}^k + \alpha\bvec{A}^t\bvec{Mb},
\end{align} 
where $\alpha = \sqrt{3}$ for DiFIGHT, and $\alpha = 3$ for MoDiFIGHT.

Furthermore, if $\max_i\sum_{j=1}^L \omega_{j}a_{ji}<1/\alpha$, or, $\max_j \omega_{j}<1/\alpha$, then the matrix $\alpha\bvec{\bvec{A}^t\Omega}$ is stable{\footnote{A matrix is said to be stable if it has spectral radius less than unity.}} and consequently, there is at least one limit point of the sequence $\{\bvec{h}^k\}_{k\ge 0}$ and for any such limit point $\bvec{h}$, the following holds: \begin{align}
\label{eq:DiFIGHT-MoDiFIGHT-convergence}
\bvec{h}\preccurlyeq \alpha\left(\bvec{I}-\alpha\bvec{A}^t\bvec{\Omega}\right)^{-1}\bvec{A}^t\bvec{Mb}.
\end{align}
\end{thm}
\begin{proof}
The proof is presented in Appendix~\ref{sec:appendix-proof-thm-diffusion-iht}.
\end{proof}
%
%\begin{thm}
%\label{thm:convergence-difhtp}
%Under the RPDH assumption, at any iteration $n$, the iterate produced by DIFHTP satisfies the following inequality:
%\begin{align}
%\label{eq:DIFHTP-vector-inequality}
%\bvec{h}^{n+1} & \bvec{\le} \bvec{D}_1\bvec{A}^t\bvec{\Omega h}^n+\bvec{D}_1\bvec{A}^t\bvec{M}\bvec{b}+\bvec{D}_2\bvec{M}\bvec{b}_2.
%\end{align}
%
%Moreover, if $\max_i\frac{\sum_{j=1}^L \omega_j a_{ji}}{\sqrt{1-\rho_i^2}}<\frac{1}{\sqrt{2}}$, or, $ \frac{\max_i\omega_i}{\sqrt{1-\left(\max_i \rho_i\right)^2}}<\frac{1}{\sqrt{2}}$, then \begin{align}
%\label{eq:DIFHTP-convergence}
%\limsup_{n\to \infty}\bvec{h}^n\bvec{\le} \left(\bvec{I}-\bvec{D}_1\bvec{A}^t\bvec{\Omega}\right)^{-1}\left(\bvec{D}_1\bvec{A}^t\bvec{Mb} + \bvec{D}_2\bvec{M}\bvec{b}_2\right).
%\end{align}
%\end{thm}
%
%specifically, \begin{itemize}
%\item ATC (adapt then combine): In this case $\bvec{A_1=I}$. Thus a sufficient condition for successful joint recovery becomes \begin{align}
%\label{eq:vector-inequality-atc}
%\lambda_{max}(\bvec{A}_2^t\bvec{D}^2\bvec{A}_2)<1/4
%\end{align} 
%\item CTA (combine then adapt): In this case $\bvec{A_2=I}$. Thus a sufficient condition for successful joint recovery becomes  \begin{align}
%\label{eq:vector-inequality-cta}
%\lambda_{max}(\bvec{A}_1^t\bvec{D}^2\bvec{A}_1)<1/4
%\end{align} 
%\end{itemize}
%\begin{proof}
%The proof is presented in Appendix~\ref{sec:appendix-DIFHTP-analysis}.
%\end{proof}
%
\subsubsection{Randomized algorithms}
We now present the main results regarding the convergence of the randomized DiFIGHT and DIFHTP algorithms for different random selection strategies. For this purpose, we introduce the diagonal matrix $\bvec{P}=\diag{\pi_1,\cdots,\ \pi_L}$ with diagonal entries $\pi_i,\ 1\le i\le L$ which were defined in Section~\ref{sec:communication-complexity-discussion} and were evaluated explicitly in Table~\ref{tab:communication-complexities} for uniform distribution for selection of group of nodes.
\begin{thm}
\label{thm:convergence-randomized-diffusion-algorithms}
Under the RPDH condition and the randomly persistent network assumption, the iterates of the randomized DiFIGHT as well as randomized MoDiFIGHT satisfy the following inequalities at time step $n$:
\begin{align}
\label{eq:randomized-DiFIGHT-MoDiFIGHT-evolution-thm}
\expect{\bvec{h}^{k+1}}\preccurlyeq (\bvec{I}-\bvec{P}+\alpha\bvec{PA}^t\bvec{\Omega})\expect{\bvec{h}^k}+\alpha \bvec{PA}^t\bvec{Mb},
\end{align}
where $\alpha = \sqrt{3}$ for DiFIGHT, and $\alpha = 3$ for MoDiFIGHT.
Consequently, under the condition, $\max_{i}\sum_{j=1}^L a_{ji}\omega_j<1/\alpha$ or, $\max_{i} \omega_i<1/\alpha$, there is at least one limit point of the sequence $\{\expect{\bvec{h}^k}\}_{k\ge 0}$, and for any such limit point $\bvec{h}$, the following bound is satisfied:
\begin{align}
\label{eq:randomized-DiFIGHT-MoDiFIGHT-convergence}
\bvec{h}\preccurlyeq \alpha\left(\bvec{I}-\alpha\bvec{A}^t\bvec{\Omega}\right)^{-1}\bvec{A}^t\bvec{Mb}.
\end{align}
Furthermore, if $\bvec{x}^\star$ is a stationary point of all the functions $f_i,\ i=1,\cdots,L$, under the above specified condition, we have, for $i=1,\cdots,L$, $\bvec{x}_i^k\to\bvec{x}^\star$ a.s.
%\item For randomized DIFHTP:
%\begin{align}
%\expect{\bvec{h}^{n+1}} & \bvec{\le } \bvec{D}_1(\bvec{I}-\bvec{P}+\bvec{PA})\bvec{\Omega}\left(\expect{\bvec{h}^n}+\bvec{\Omega}^{-1}\bvec{Mb}\right)\nonumber\\
%\label{eq:randomized-DIFHTP-evolution-thm}
%\ & + \bvec{D}_2\bvec{M b}_2.
%\end{align}
%\end{enumerate}
\end{thm}
\begin{proof}
The proof is supplied in Appendix~\ref{sec:appendix-randomized-algorithms-analysis}.
\end{proof}
%
%\subsection{Example two node network}
%Let us consider a two node network to analyze the performance of the DiFIGHT and DIFHTP algorithms in detail. We consider the following combination matrix \begin{align*}
%\bvec{A} = \begin{bmatrix}
%a & 1-a \\
%1-a & a
%\end{bmatrix}
%\end{align*}  
%where $a\in [0,1]$. Then, for DiFIGHT, we have \begin{align*}
%\bvec{B} =2\begin{bmatrix}
%\rho_1 a & \rho_2 (1-a)\\
%\rho_1 (1-a) & \rho_2 a
%\end{bmatrix}
%\end{align*}
\section{Simulation results}
\label{sec:simulation-results}
\subsection{Simulation setup}
\label{sec:simulation-setup}
In this section we perform numerical study of the DiFIGHT and MoDiFIGHT algorithms along with the Consensus IHT algorithm, which is similar to DiFIGHT, only with the exception that the nodes first exchange the estimates and then use their individual gradient vector for the hard thresholding update. We also plot the performance of the non-cooperative IHT, where all nodes simply run their own algorithm and do not communicate with each other, as well as the performance of the centralized IHT algorithm which is executed using the measurements available throughout all the nodes in the network. In all the experiments, the unknown vector $\bvec{x}$ has a fixed dimension $n=200$, and sparsity $K=10$. The indices for support of $\bvec{x}$ is sampled uniformly from $1,2,\cdots,\ n$, and then the values at those indices are generated according to $\mathcal{N}(0,1)$ distribution. We consider networks with $L = 10,\ 15$ nodes for our experiments. For each $L$, the network is generated using Erd\H{o}s-Reyni model where there is a link between two nodes is with probability $p$, and not generated with probability $1-p$. $p=\frac{\ln L}{L}$ is selected to get a connected graph with high probability~\cite{erdds1959random}. The generated graph is checked for full connectivity using depth-first search algorithm, and the process is continued until a connected graph is obtained. The adjacency matrix of the graph thus obtained, is normalized to make it left stochastic, and is used as the combination matrix $\bvec{A}$. For each node, the measurement model is taken to be the noiseless linear measurement model, where the node $v$ has a measurement $\bvec{y}_v$ available with it which is obtained from an unknown signal $\bvec{x}^\star$ via the linear transformation $\bvec{y}_v = \bvec{\Phi}_v\bvec{x}^\star$, where the $\bvec{\Phi}_v$ is a $m\times n$ measurement matrix that is generated with entries sampled from i.i.d.~$\mathcal{N}(0,1/m)$ distribution. All the algorithms are run for $100$ instances and for each instance independent copies of the target vector $\bvec{x}^\star$, and the measurement matrix $\bvec{\Phi}_v$ are generated. However, for a particular $L$, the underlying network is kept the same throughout all these instances. 
\subsection{Probability of recovery performance}
\label{sec:simulation-probability-recovery}
\begin{figure}[t!]
%
%\begin{subfigure}{.5\textwidth}
%\centering
%\includegraphics[height= 2.5in, width = 3.5in]{Recovery_probability_vs_total_no_measurements_N=200_K=10_L=5_determinstic_strategies}
%\caption{$L=5$}
%\end{subfigure}
%
\begin{subfigure}{.5\textwidth}
\centering
\includegraphics[height= 2.5in, width = 3.5in]{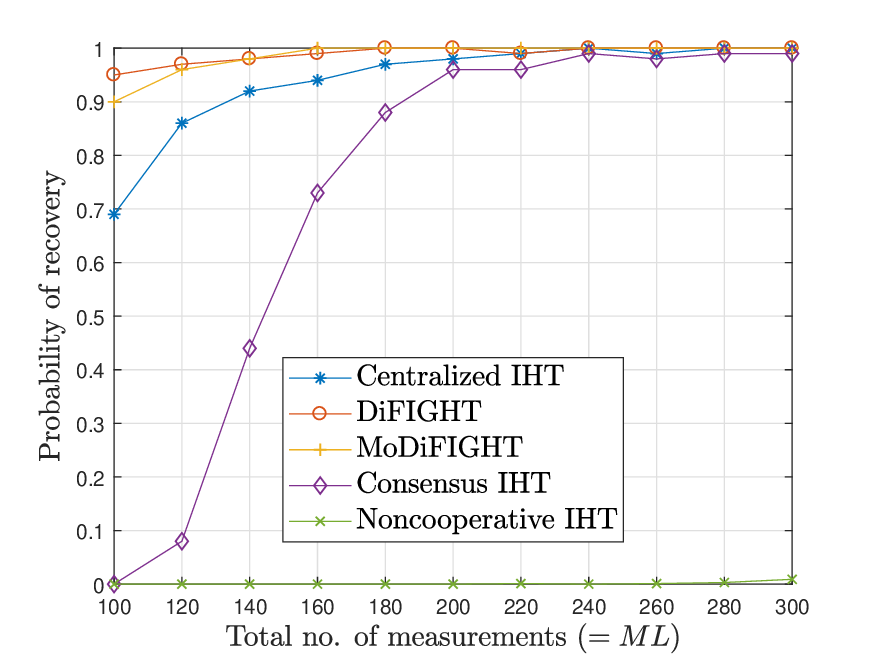}
\caption{$L=10$}
\label{fig:prob-recovery-perfo-deterministic-L=10}
\end{subfigure}
\begin{subfigure}{.5\textwidth}
\centering
\includegraphics[height= 2.5in, width = 3.5in]{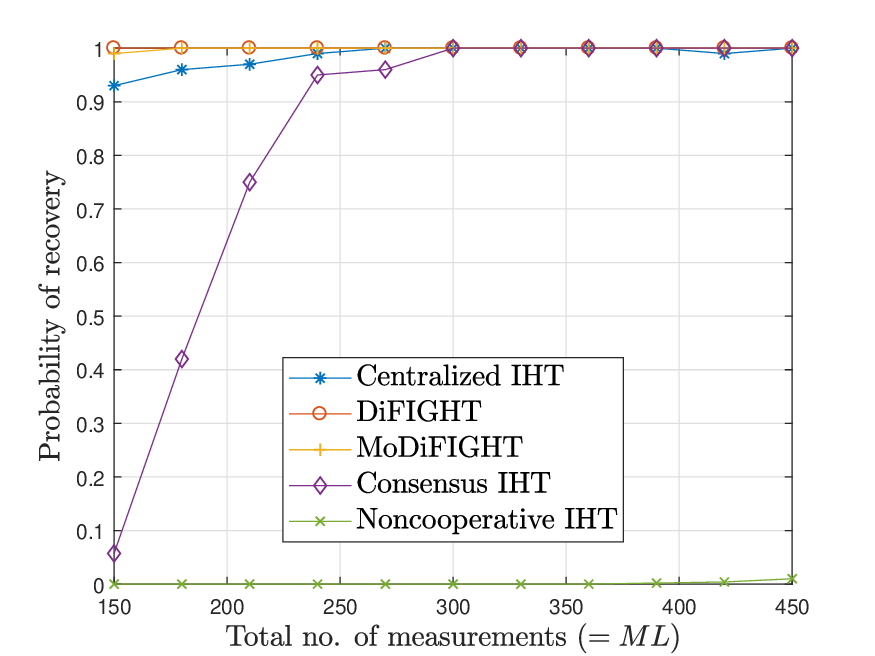}
\caption{$L=15$}
\label{fig:prob-recovery-perfo-deterministic-L=15}
\end{subfigure}
\caption{Probability of recovery vs number of measurements using all nodes in the network}
\label{fig:probability-recovery-perfo-deterministic}
\end{figure}
\begin{figure}[t!]
%
%\begin{subfigure}{.5\textwidth}
%\centering
%\includegraphics[height= 2.5in, width = 3.5in]{Recovery_probability_vs_total_no_measurements_N=200_K=10_L=5_randomized_strategies}
%\caption{$L=5$}
%\end{subfigure}
%
\begin{subfigure}{.5\textwidth}
\centering
\includegraphics[height= 2.5in, width = 3.5in]{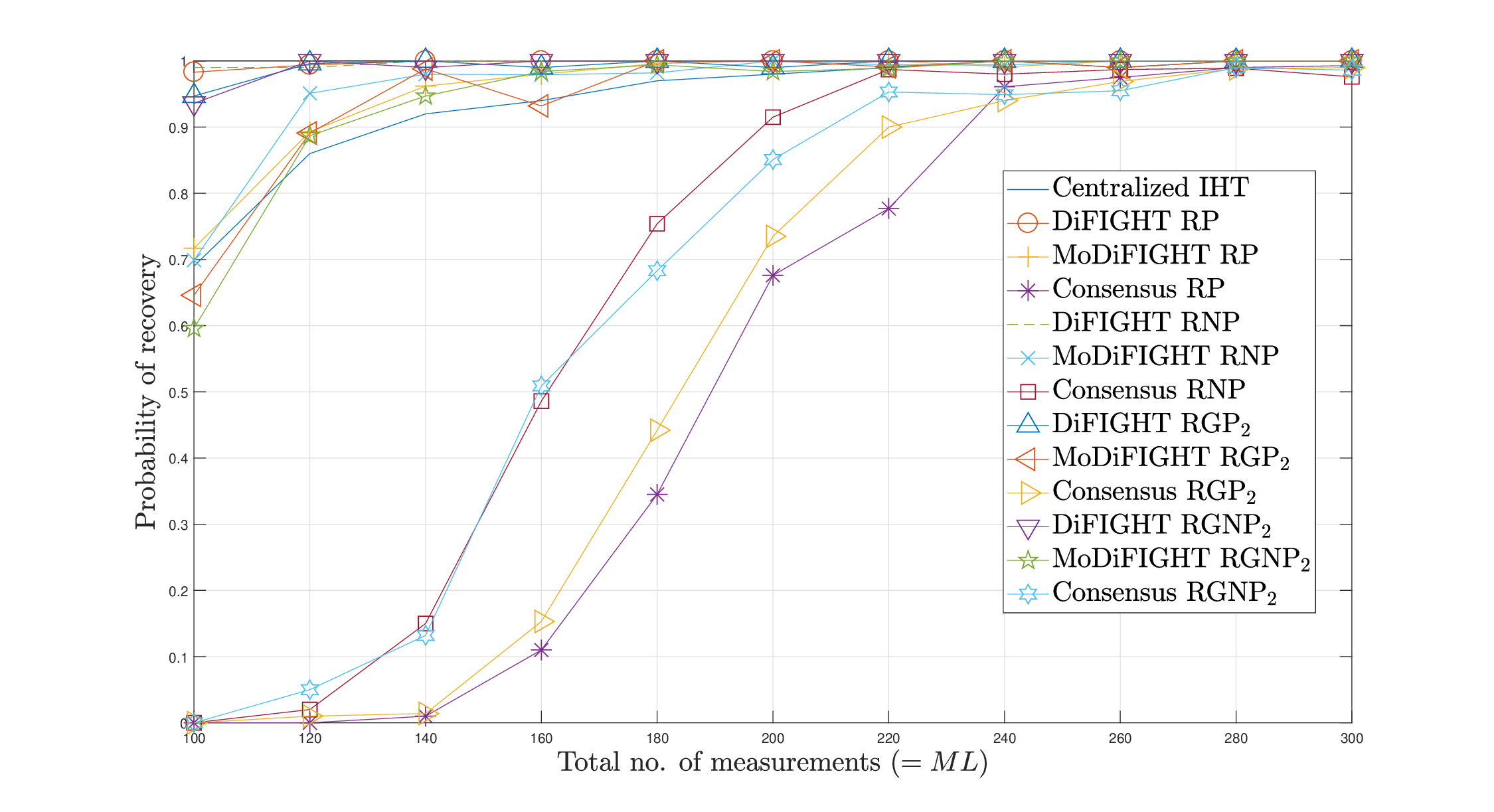}
\caption{$L=10$}
\label{fig:prob-recovery-perfo-random-L=10}
\end{subfigure}
\begin{subfigure}{.5\textwidth}
\centering
\includegraphics[height= 2.5in, width = 3.5in]{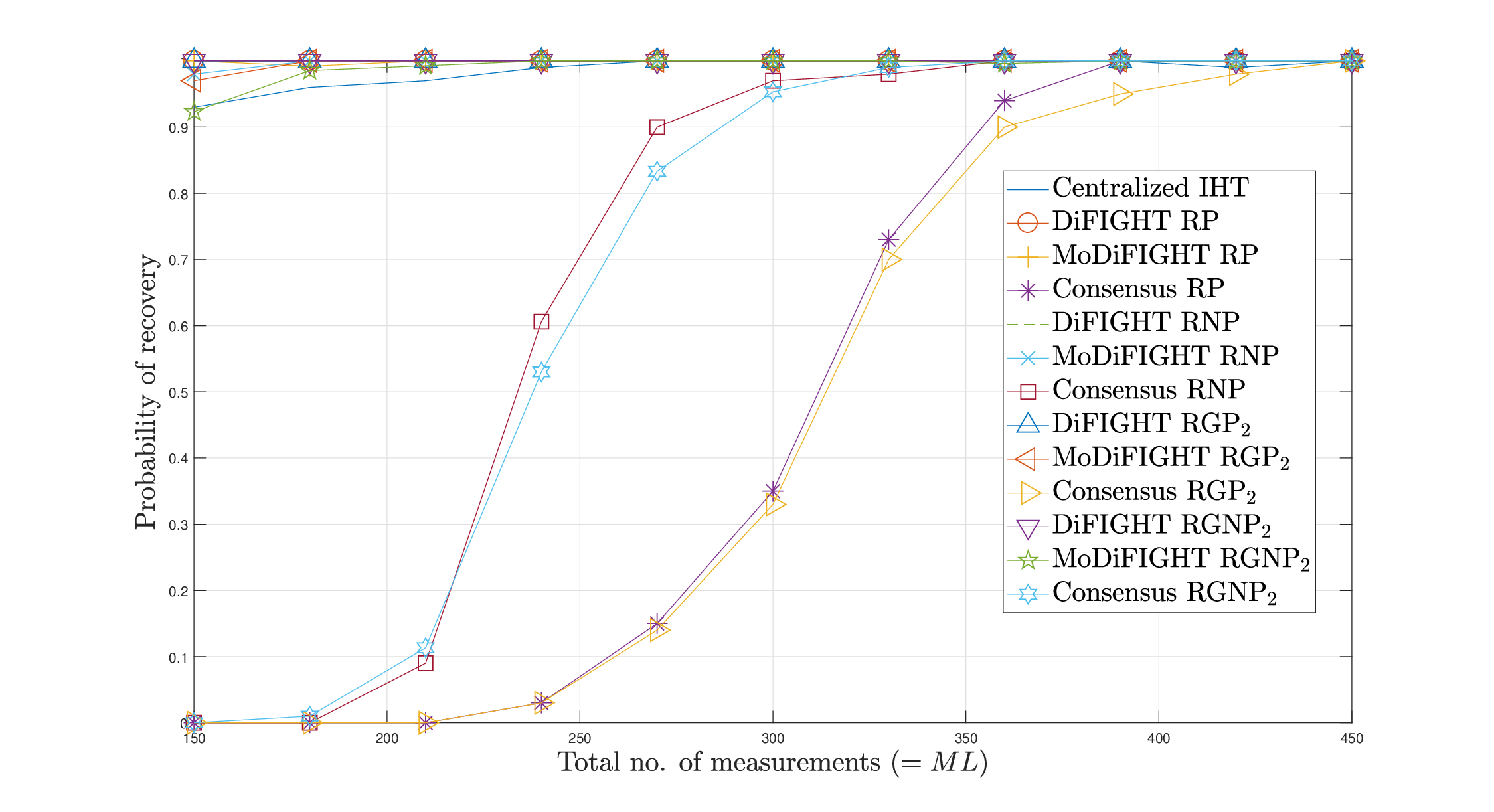}
\caption{$L=15$}
\label{fig:prob-recovery-perfo-random-L=15}
\end{subfigure}
\caption{Probability of recovery vs number of measurements for randomized node selection strategies}
\label{fig:probability-recovery-perfo-randomized}
\end{figure}
In this experiment we plot the probability with which the different algorithms recover the unknown signal $\bvec{x}^\star$. Note that in the horizontal axis we have the \emph{total} number of measurements available throughout the network. The performance of the centralized IHT is evaluated taking all these measurements. On the other hand, for the distributed algorithm, each node has access to considerably smaller number of measurements, for example, if the total number of measurements is $150$ and the network size is $L=10$, then each node of the network has access to only $15$ measurements. To calculate the probability of recovery, we calculate the number of instances (out of the $100$ instances) in which an algorithm has a ``successful'' recovery, where a successful recovery is quantified as follows: 1) for the centralized algorithm, we call an instance or run of an algorithm successful if the estimate $\hat{\bvec{x}}$ produced by the algorithm satisfies $\frac{\opnorm{\hat{\bvec{x}} - \bvec{x}^\star}_2^2}{\opnorm{\bvec{x}^\star}_2^2}<10^{-4} $. 2) For the distributed case, an instance or a run of an algorithm is called successful if the algorithm produces estimates $\{\hat{\bvec{x}}_i\}_{1\le i\le L}$ in all the nodes of the network, such that $\frac{\sum_{j=1}^L \opnorm{\hat{\bvec{x}}_j - \bvec{x}^\star}_2^2}{L\opnorm{\bvec{x}^\star}_2^2}<10^{-4}$. 

\paragraph{\textbf{Performance of the deterministic algorithms}} The figure~\ref{fig:probability-recovery-perfo-deterministic} compares the probability of recoveries of the different algorithms considered in this paper. From this figure, one can appreciate the substantial amount performance gain offered by the DiFIGHT and MoDiFIGHT algorithms over the consensus IHT algorithm and even over that of the centralized algorithm. This gain can be explained using the fact that these diffusion algorithm leverage the diversity offered by the different gradient vectors gathered from the neighborhood of a node. We also see that the distributed algorithms require very small number of measurements for successful recovery compared to the standalone algorithms, as is exemplified by the abysmal performance of the non-cooperative IHT algorithm. For example, from the Figure~\ref{fig:prob-recovery-perfo-deterministic-L=15}, we see that all the distributed algorithms have recovery probability $1$ after $m$ crosses $20$, whereas the recovery probability of the non-cooperative algorithm is almost $0$ even when $m$ is close to $30$. We also observe that the performance of the DiFIGHT and MoDiFIGHT algorithms are very close, with the latter exhibiting slightly poorer performance than the former only for small $m$ (~$10-12$ ).
\paragraph{\textbf{Performance under the randomized strategies}} The figure~\ref{fig:probability-recovery-perfo-randomized} demonstrates the relative performances of the DiFIGHT, MoDiFIGHT and consensus IHT algorithms for the different randomized node selection strategies which were discussed in Section~\ref{sec:randomized-strategies}. We have used group size $2$ for the experiment. From the plots we observe that the RP strategy performs the worst among all the four strategies, which is expected as only one node at a time is selected in this strategy. The RGP$_r$ strategy is slightly better than the RP strategy as a few nodes are selected. But the best strategies are seen to be the RNP and RGP$_r$ strategies as in both these strategies many neighboring nodes are selected at a time, which elevates the eprformance of the distributed algorithms, especially in dense networks.
\subsection{Mean square deviation performance}
\label{sec:simulation-msd-performance}
\begin{figure}[t!]
%
%\begin{subfigure}{.5\textwidth}
%\centering
%\includegraphics[height= 2.5in, width = 3.5in]{MSD_vs_iteration_no_N=200_K=10_L=5_M=30_deterministic_strategies}
%\caption{$M=30,\ L=5$}
%\end{subfigure}
%
\begin{subfigure}{.5\textwidth}
\centering
\includegraphics[height= 2.5in, width = 3.5in]{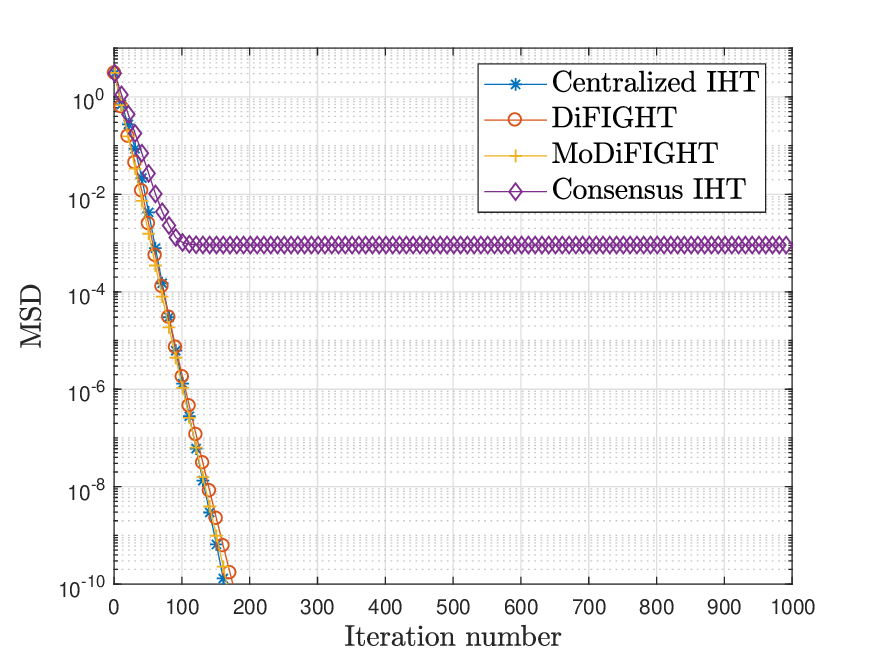}
\caption{$M=30,\ L=10$}
\end{subfigure}
\begin{subfigure}{.5\textwidth}
\centering
\includegraphics[height= 2.5in, width = 3.5in]{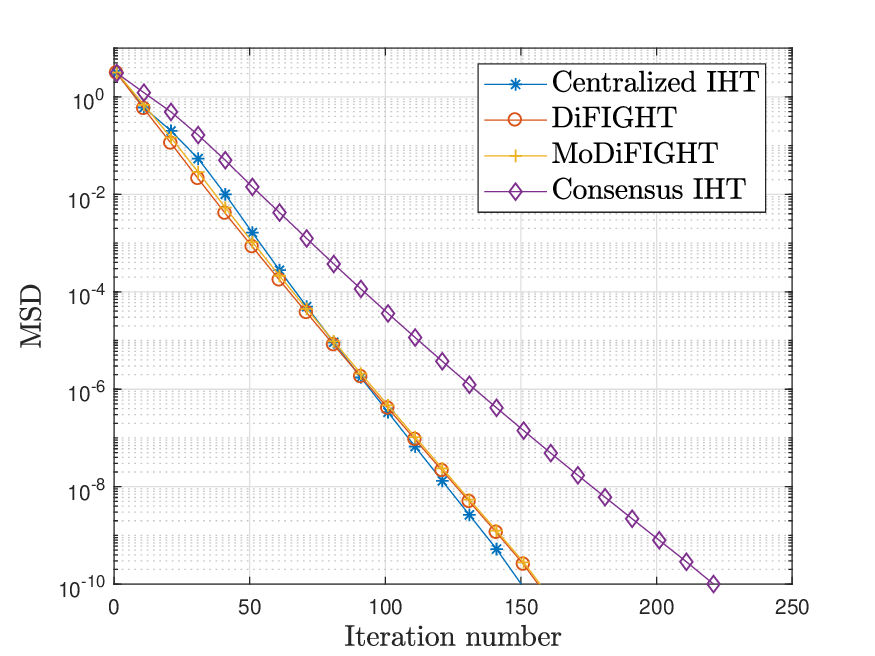}
\caption{$M=30,\ L=15$}
\end{subfigure}
\caption{MSD vs iteration number when all nodes are used, $M=30$}
\label{fig:msd-perfo-deterministic}
\end{figure}
\begin{figure}[t!]
%
%\begin{subfigure}{.5\textwidth}
%\centering
%\includegraphics[height= 2.5in, width = 3.5in]{MSD_vs_iteration_number_N=200_K=10_L=5_M=30_randomized_strategies}
%\caption{$M=30,\ L=5$}
%\end{subfigure}
%
\begin{subfigure}{.5\textwidth}
\centering
\includegraphics[height= 2.5in, width = 3.5in]{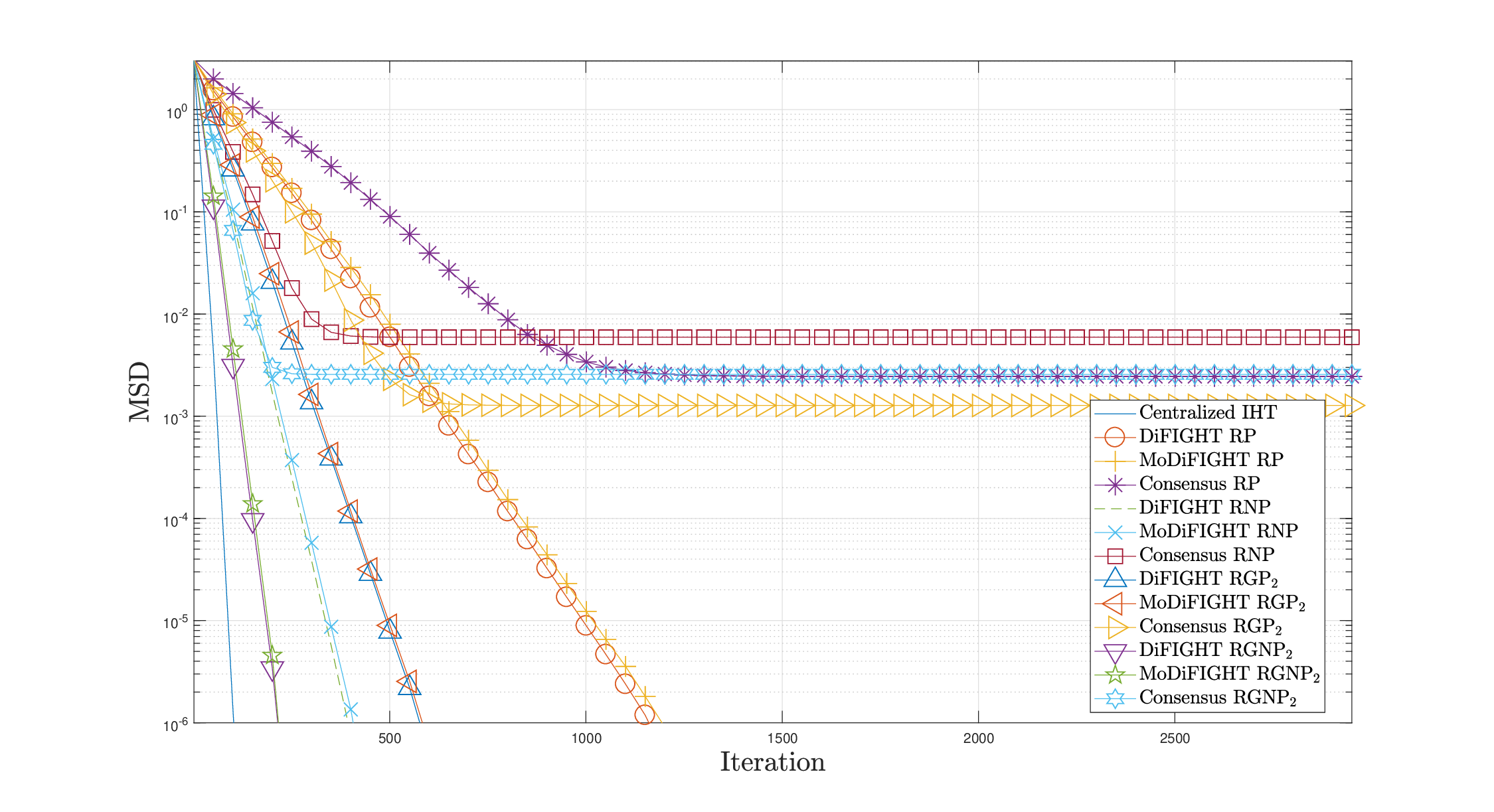}
\caption{$M=30,\ L=10$}
\end{subfigure}
\begin{subfigure}{.5\textwidth}
\centering
\includegraphics[height= 2.5in, width = 3.5in]{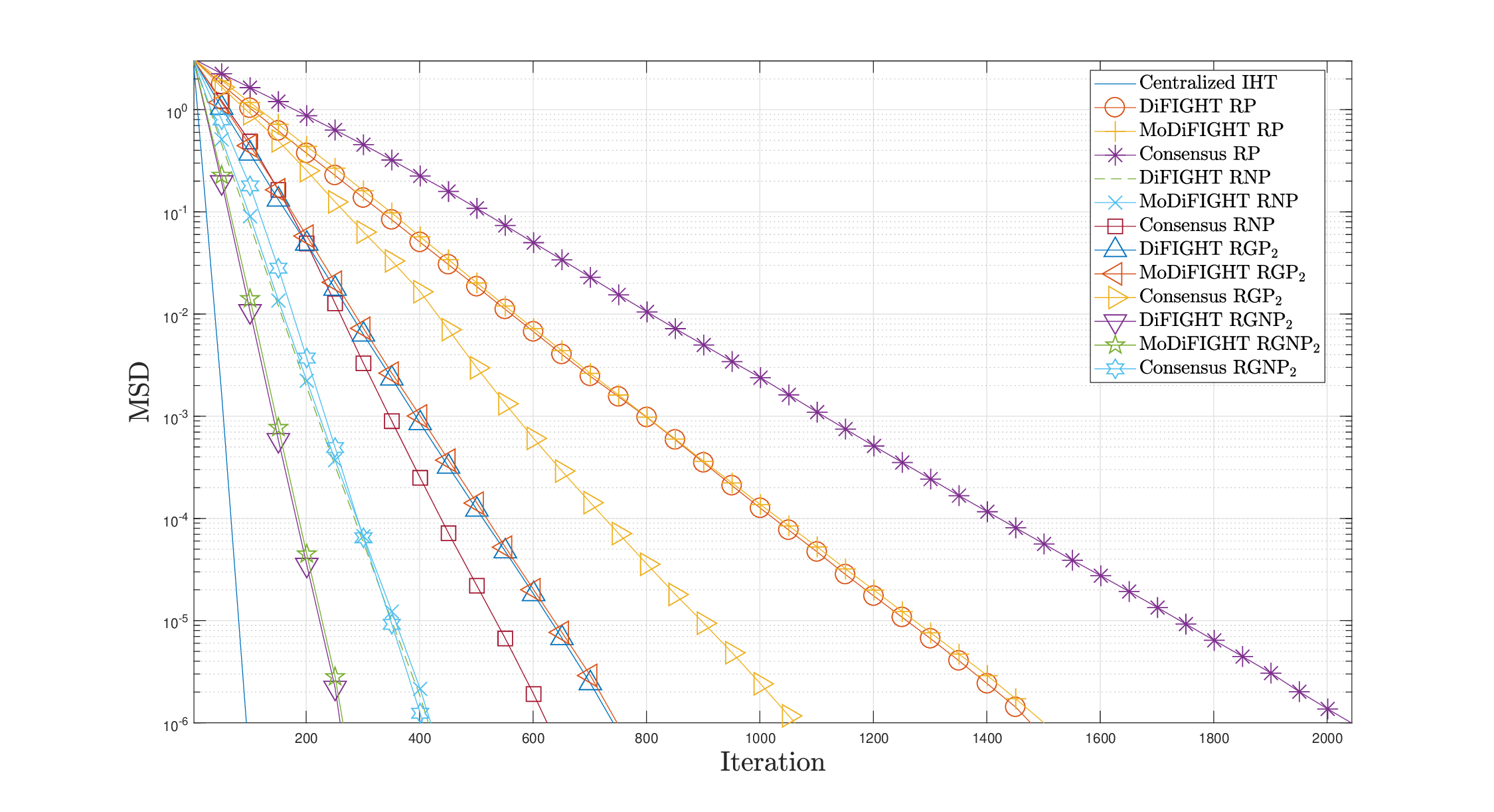}
\caption{$M=30,\ L=15$}
\end{subfigure}
\caption{MSD vs iteration number for randomized strategies, $M=30$}
\label{fig:msd-perfo-randomized}
\end{figure}
\paragraph{\textbf{Deterministic algorithms}}
We see that for $m=30,\ L=10$, although the consensus IHT converges too early to result in high mean square deviation (MSD), the DiFIGHT and MoDiFIGHT algorithms continue to have decreasing MSD and the rate of the algorithms match with that of the centralized IHT. For $m=30,\ L=15$, all the algorithms show good MSD performance with the DiFIGHT and MoDiFIGHT showing superior convergence rate.
\paragraph{\textbf{Randomized strategies}}
For both $m=30,\ L=10$ and $m=30,\ L=15$, we again see that for all the randomized strategies the consensus IHT has poorer MSD performance compared to the DiFIGHT and MoDiFIGHT algorithms. However, we observe that the convergence rate of each algorithm depends on the randomized strategy used, specifically according to decreasing convergence rate the strategies are seen to be ranked as: RGNP$_r$ $>$ RNP $>$ RGP$_r$ $>$ RP. This corroborates the intuition that cooperation in the network increases convergence speed.
%
%\subsection{Communication complexity performance}
%\label{sec:simulation-communication-complexity-performance}
%%
%
%%
\appendices
\section{Proof of Lemma~\ref{lem:irreducibility} }
\label{sec:appendix-proof-irreducibility-lemma}
The key observation is that as the graph $\mathcal{G}$ is connected, the associated weight matrix $\bvec{X}$ is irreducible. We now prove the three claims as below:\begin{enumerate}
\item Since $\bvec{X}$ is irreducible, this is trivially true since for any two indices $i,j$, $\exists $ a positive integer $t_{ji}$, such that $((\bvec{X}^t)^{t_{ji}})_{ij} = ((\bvec{X})^{t_{ji}})_{ji}>0$.
% However, this is trivially true as $\bvec{X}$ is irreducible, and hence $\exists k_{ji}$, a positive integer, such that $(\bvec{X}^{k_{ji}})_{ji}>0$, or equivalently, $((\bvec{X}^t)^{k_{ji}})_{ij}>0$. Thus, $\bvec{X}^t$ is irreducible.
\item Since each element of $\bvec{D}_1,\bvec{D}_2$ is strictly positive, $\exists \alpha>0$ such that $d_{1i},d_{2i}\ge \alpha,\ \forall i=1,\cdots,L$. Therefore, the $(i,j)^{\mathrm{th}}$ element of $\bvec{D}_1\bvec{X}\bvec{D}_2$ is $d_{1,i}(\bvec{X})_{ij}d_{2,j}\ge \alpha^2 (\bvec{X})_{ij}$. Since $\bvec{X}$ is irreducible, for any $1\le i,j\le L$, there exists a positive integer $ t_{ij}$ such that $(\bvec{X}^{t_{ij}})_{ij}>0$. Therefore, $((\bvec{D}_1\bvec{X}\bvec{D}_2)^{t_{ij}})_{ij}\ge \alpha^2 (\bvec{X}^{t_{ij}})_{ij}>0$, which establishes the claim.
\item Observe that for any $1\le i,j\le L$, $(\bvec{X} + \bvec{M})_{ij}\ge (\bvec{X})_{ij}$ and that $\bvec{X}$ is irreducible. Therefore, $\bvec{X} + \bvec{M}$ is irreducible.
\end{enumerate}
\section{Proof of Lemma~\ref{lem:convergence-linear-inequality} }
\label{sec:appendix-proof-lemma-convergence-linear-inequality }
First note that since the sequence $\{\bvec{u}^k\}_{k\ge 0}$, as well as the matrix $\bvec{B}$ and vector $\bvec{b}$ are non-negative, one finds that for each $k\ge 0$, $\bvec{u}^{k+1}\preccurlyeq \bvec{B}^{k+1}\bvec{u}^0 + \sum_{j=0}^k \bvec{B}^j \bvec{b}\preccurlyeq(\bvec{I}-\bvec{B})^{-1}(\bvec{u}^0 + \bvec{b})$, where we have used the fact that the Neumann series $\sum_{j=0}^\infty \bvec{B}^j$ converges to $(\bvec{I}-\bvec{B})^{-1}$ as the matrix $\bvec{B}$ is stable. Thus the sequence $\{\bvec{u}^k\}_{k\ge 0}$is non-negative as well as upper bounded, which ensures, by the Bolzano-Weierstrass theorem that there is at least one limit point of the sequence $\{\bvec{u}^k\}_{k\ge 0}$. Then, if  $\bvec{v}$ is a limit point of $\{\bvec{u}^k\}_{k\ge 0}, $ by definition, there exists a strictly increasing sequence of non-negative integers $k_j\ge 0, j=1,2,\cdots$, such that $\lim_{j\to \infty}\bvec{u}^{k_j}=\bvec{v}$ (see~\cite[Chapter 3]{rudin1964principles}). Consequently, we obtain, $\bvec{v}=\lim_{j\to \infty}\bvec{u}^{k_j}\preccurlyeq\lim_{j\to \infty}\bvec{B}^{k_j}\bvec{u}_0 + \sum_{i=0}^\infty \bvec{B}^i \bvec{b} = (\bvec{I}-\bvec{B})^{-1}\bvec{b}$, which concludes the proof.
\section{Proof of Theorem~\ref{thm:convergence-DiFIGHT} }
\label{sec:appendix-proof-thm-diffusion-iht}
To perform the analysis for both DiFIGHT and MoDiFIGHT we employ the technique of analysis of Theorem 3.5 of~\cite{foucart2011hard} and extend it to the distributed case.

We fix any $i=1,\cdots,L$. From the description of the DiFIGHT algorithm in Table~\ref{tab:DiFIGHT-MoDiFIGHT}, using the expression for $\bvec{\psi}_j^{k+1}$ and writing $\bvec{e}_j^k = \bvec{x}_j^k-\bvec{x}^\star-\mu_j\nabla f_j(\bvec{x}_j^K)-\nabla f_j(\bvec{x}^\star)$, one can derive using triangle inequality that, 
\begin{align}
\lefteqn{\opnorm{(\bvec{x}^{k+1}_i-\bvec{x}^\star)_{\Lambda_i^{k+1}}}_2 } & &\nonumber\\
\ & \le \sum_{j=1}^L a_{ji}\opnorm{\left(\bvec{e}_j^{k+1}\right)_{\Lambda_i^{k+1}}}_2 + \sum_{j=1}^L a_{ji}\mu_j \opnorm{(\nabla f_j(\bvec{x}^\star))_{\Lambda_{i}^{k+1}}}_2\nonumber\\
\label{eq:k+1-step-first-inequality}
\ & \le \sum_{j=1}^L a_{ji} \left(\omega_j \opnorm{\bvec{x}^{k}_j-\bvec{x}^\star}_2 +\mu_j\opnorm{\nabla_{K} f_j(\bvec{x}^\star)}_2 \right)\nonumber\\
\ & \le \sum_{j=1}^L a_{ji} \left(\omega_j \opnorm{\bvec{x}^{k}_j-\bvec{x}^\star}_2 +\mu_jb_j \right)
\end{align}
where the last step used the fact that $\opnorm{\nabla_{K} f_j(\bvec{x}^\star)}_2 \le \opnorm{\nabla_{2K} f_j(\bvec{x}^\star)}_2=b_j $ and inequality~\eqref{eq:rpdh-inequality-result2} of Lemma~\ref{lem:RPDH-inprod-inequality} which used the fact that $\abs{\Lambda_i^k\cup \Lambda\cup \Lambda_i^{k+1}}\le 3K$.
%\begin{align}
%\bvec{x}^{k+1}_i & = H_K\left(\sum_{j=1}^L a_{ji}\left(\bvec{x}_j^k - \mu_j\nabla f_j(\bvec{x}_j^k)\right)\right).
%\end{align} 
%%Let us define, $\bvec{A}=\bvec{A}_1\bvec{A}_2$, and define, $b_{ijk}=a_{1,kj}a_{2,ji},\ i,j,k\in \{1,2,\cdots,\ L\}$.
%Then, from the definition of the operator $H_K$, we obtain, $\forall\ i=1,2,\cdots,\ L$,
On the other hand, since the support of $\bvec{x}_i^{k+1}$ is $\Lambda_i^{k+1}$, one obtains,
\begin{align}
\opnorm{\left(\sum_{j=1}^L a_{ji}\bvec{\psi}_j^{k+1}\right)_{\Lambda_i^{k+1}}}_2 & \ge \opnorm{\left(\sum_{j=1}^L a_{ji}\bvec{\psi}_j^{k+1}\right)_{\Lambda}}_2\nonumber\\
\label{eq:comparison-inequality-preliminary}
\Leftrightarrow \opnorm{\left(\sum_{j=1}^L a_{ji}\bvec{\psi}_j^{k+1}\right)_{\Lambda_i^{k+1}\setminus \Lambda}}_2 & \ge \opnorm{\left(\sum_{j=1}^L a_{ji}\bvec{\psi}_j^{k+1}\right)_{\Lambda\setminus \Lambda1^{k+1}_i}}_2.
\end{align}
 Again, using the expression of $\bvec{\psi}_j^{k+1}$, it is easy to observe that \begin{align}
\lefteqn{\opnorm{\left(\sum_{j=1}^L a_{ji}\bvec{\psi}_j^{k+1}\right)_{\Lambda_i^{k+1}\setminus \Lambda}}_2} & & \nonumber\\
\ & \le \sum_{j=1}^L a_{ji}\opnorm{\left(\bvec{e}_j^k\right)_{\Lambda_i^{k+1}\setminus \Lambda}}_2 + \sum_{j=1}^L a_{ji}\mu_j\opnorm{(\nabla f_j(\bvec{x}^\star))_{\Lambda_i^{k+1}\setminus \Lambda}},
\end{align}
where the last step used triangle inequality and the fact that $(\bvec{x}^\star)_{\Lambda_i^{k+1}\setminus \Lambda}=\bvec{0}$. Similarly, one can find,  \begin{align}
\lefteqn{\opnorm{\left(\sum_{j=1}^L a_{ji}\bvec{\psi}_j^{k+1}\right)_{\Lambda\setminus \Lambda^{k+1}_i}}_2} & &\nonumber\\
\ & \ge \opnorm{(\bvec{x}_i^{k+1}-\bvec{x}^\star)_{\Lambda\setminus \Lambda_i^{k+1}}}_2 - \sum_{j=1}^L a_{ji}\opnorm{(\bvec{e}^{k}_{j})_{\Lambda\setminus \Lambda^{k+1}_i}}_2\nonumber\\
\ & - \sum_{j=1}^L a_{ji}\mu_j\opnorm{(\nabla f_j(\bvec{x}^\star))_{\Lambda\setminus \Lambda_i^{k+1}}}_2.
\end{align}
Therefore, from inequality~\eqref{eq:comparison-inequality-preliminary} it follows that \begin{align}
\lefteqn{\opnorm{(\bvec{x}_i^{k+1}-\bvec{x}^\star)_{\Lambda\setminus \Lambda_i^{k+1}}}_2} & &\nonumber\\
\ & \le \sum_{j=1}^L a_{ji}\left(\opnorm{(\bvec{e}^{k}_{j})_{\Lambda\setminus \Lambda^{k+1}_i}}_2 + \opnorm{(\bvec{e}^{k}_{j})_{\Lambda^{k+1}_i\setminus \Lambda}}_2\right)\nonumber\\
\ & + \sum_{j=1}^L a_{ji}\mu_j\left(\opnorm{(\nabla f_j(\bvec{x}^\star))_{\Lambda\setminus \Lambda_i^{k+1}}}_2+\opnorm{(\nabla f_j(\bvec{x}^\star))_{\Lambda_i^{k+1}\setminus\Lambda }}_2\right)\nonumber\\
\ & \le \sqrt{2}\sum_{j=1}^L a_{ji}\left(\opnorm{(\bvec{e}_j^{k})_{\Lambda_{i}^{k+1}\Delta \Lambda}}_2+\mu_j\opnorm{(\nabla f_j(\bvec{x}^\star))_{\Lambda_i^{k+1}\Delta\Lambda}}_2\right)\nonumber\\
\ & \le \sqrt{2}\sum_{j=1}^L a_{ji}\left(\omega_j\opnorm{\bvec{x}^{k}_j-\bvec{x}^\star}_2 + \mu_j\opnorm{\nabla_{2K}  f_j(\bvec{x}^\star)}_2\right)\nonumber\\
\ & = \sqrt{2}\sum_{j=1}^L a_{ji}\left(\omega_j\opnorm{\bvec{x}^{k}_j-\bvec{x}^\star}_2 + \mu_jb_j\right)
\end{align}
where the last step used inequality~\eqref{eq:rpdh-inequality-result2} of Lemma~\ref{lem:RPDH-inprod-inequality}. Therefore, it follows that for all $i=1,\cdots,\ L$ \begin{align}
\label{eq:difight-one-node-main-inequality}
\opnorm{\bvec{x}_i^{k+1}-\bvec{x}^\star}_2 & \le \sqrt{3}\sum_{j=1}^L a_{ji}\left(\omega_j\opnorm{\bvec{x}^{k}_j-\bvec{x}^\star}_2 + \mu_jb_j\right).
\end{align}
Using the definition, $\bvec{h}^k=\begin{bmatrix}
\opnorm{\bvec{x}_1^k - \bvec{x}^\star}_2 & \opnorm{\bvec{x}_2^k - \bvec{x}^\star}_2 & \cdots & \opnorm{\bvec{x}_L^k - \bvec{x}^\star}_2
\end{bmatrix}^t $. Then, it is easy to observe that the preceding set of inequalities can be collected together into the following vector inequality: \begin{align}
\label{eq:difight-main-inequality}
\bvec{h}^{k+1} & \preccurlyeq \bvec{H} \bvec{h}^k + \bvec{d},
\end{align} 
where $\bvec{H} = \sqrt{3}\bvec{A}^t\bvec{\Omega}$, and $\bvec{d} = \sqrt{3}\bvec{A}^t\bvec{M}\bvec{b}$, where $\bvec{M,\ b}$ are defined in Theorem~\ref{thm:convergence-DiFIGHT}.
Now, using lemma~\ref{lem:irreducibility}, we find that $\bvec{H}$ is an irreducible matrix. Then according to the Perron-Frobenius theorem~\ref{thm:perron-frobenius}, the maximum eigenvalue (according to absolute value) of $\bvec{H}$ satisfies $r\le\sqrt{3}\max_i\sum_{j=1}^L \omega_{j} a_{ji}$. By imposing the restriction $\max_i\sum_{j=1}^L \omega_{j} a_{ji}<1/\sqrt{3}$, we see from (4) of Theorem~\ref{thm:perron-frobenius} that the matrix $\bvec{H}$ is stable, and consequently, applying Lemma~\ref{lem:convergence-linear-inequality}, one finds that $\bvec{h}\bvec{\le}(\bvec{I}-\bvec{H})^{-1}\bvec{d}$, where $\bvec{h}$ is any limit point of the sequence $\{\bvec{h}^k\}_{k\ge 0}$. Furthermore, $\max_{j}\omega_{j}<1/\sqrt{3}$ ensures that $\max_i\sum_{j=1}^L \omega_{j} a_{ji}<1/\sqrt{3}$, which is a weaker sufficient condition for the stability of matrix $\bvec{H}$, that does not require the explicit knowledge of the combination matrix $\bvec{A}$.

In order to derive an evolution inequality for the MoDiFIGHT algorithm, we first note that since $\bvec{\psi}^k_j = H_K\left(\bvec{x}_j^k - \mu_j\nabla f_j(\bvec{x}_j^k)\right),\ j=1,2,\cdots,\ L$, one can readily use the analysis of Theorem 3.5 of~\cite{foucart2011hard}, as above, restricted to a single node to derive the following inequality for all nodes $j=1,\cdots,\ L$: \begin{align}
\label{eq:modifight-inequality-1}
\opnorm{\bvec{\psi}_j^k-\bvec{x}^\star}_2 & \le \sqrt{3}\omega_j\opnorm{\bvec{x}_j^k - \bvec{x}^\star}_2 + \sqrt{3}b_j.
\end{align} 

On the other hand, since $\bvec{x}_i^{k+1}=H_K\left(\sum_{j=1}^L a_{ji}\bvec{\psi}_j^k\right),$ using the analysis similar to the one for DiFIGHT for one node,  we can find that \begin{align}
\opnorm{\left(\bvec{x}_i^{k+1} - \bvec{x}^\star\right)_{\Lambda_i^{k+1}}}_2 & \le  \sum_{j=1}^L a_{ji}\opnorm{\left(\bvec{\psi}_j^{k} - \bvec{x}^\star\right)_{\Lambda_i^{k+1}}}_2,\nonumber\\
\opnorm{\left(\bvec{x}_i^{k+1} - \bvec{x}^\star\right)_{\Lambda\setminus \Lambda_i^{k+1}}}_2 & \le \sqrt{2}\sum_{j=1}^L a_{ji}\opnorm{\left(\bvec{\psi}_j^{k}-\bvec{x}^\star\right)_{\Lambda_i^{k+1}\Delta \Lambda} }\nonumber\\
\label{eq:modifight-inequality-2}
\implies \opnorm{\bvec{x}_i^{k+1} - \bvec{x}^\star}_2 & \le \sqrt{3}\sum_{j=1}^L a_{ji}\opnorm{\bvec{\psi}_j^{k}-\bvec{x}^\star }.
\end{align}
 Therefore the inequalities~\eqref{eq:modifight-inequality-1} and~\eqref{eq:modifight-inequality-2} together yield the following 
main inequality governing the evolution of $\opnorm{\bvec{x}_i^{k+1} - \bvec{x}^\star}_2$ for node $i$ for MoDiFIGHT: 
\begin{align}
\lefteqn{\opnorm{\bvec{x}^{k+1}_i - \bvec{x}^\star}_2} & &\nonumber\\
\label{eq:modifight-one-node-main-inequality}
\ & \le 3\sum_{j=1}^L a_{ji}\omega_j\opnorm{\bvec{x}^k_j - \bvec{x}^\star}_2 + 3\sum_{j=1}^L a_{ji} \mu_j\opnorm{\nabla_{2K} f_j(\bvec{x}^\star)}_2
\end{align}
We note that the inequality~\eqref{eq:modifight-one-node-main-inequality} is essentially the same as the inequality~\eqref{eq:difight-one-node-main-inequality} only with a factor of $3$ instead of $\sqrt{3}$ at the front. Thus, using similar analysis as present in the part of the analysis of DiFIGHT after inequality~\eqref{eq:difight-one-node-main-inequality}, we arrive at the following vector inequality: 
\begin{align}
\label{eq:modifight-main-inequality}
\bvec{h}^{k+1} & \preccurlyeq \sqrt{3}\bvec{H} \bvec{h}^k + 2\bvec{d},
\end{align}
where $\bvec{H},\ \bvec{d}$ are defined as in the analysis of the DiFIGHT algorithm. Also, we find that a sufficient condition for the right hand side of the above inequality to converge is $\max_i\sum_{j=1}^L a_{ji}\omega_j<1/3$, or a weaker condition $\max_j \omega_j <1/3$.  
\section{Proof of Theorem~\ref{thm:convergence-randomized-diffusion-algorithms} }
\label{sec:appendix-randomized-algorithms-analysis}
To carry out the proof, let us first consider a node $v\in G_k$, where $G_k$ is the group of nodes chosen at time $k$. For both the randomized DiFIGHT and MoDiFIGHT algorithms, derivation of the evolution of the norm of the error $\opnorm{\bvec{x}_v^{k+1}-\bvec{x}^\star}_2$ in terms of $\opnorm{\bvec{x}_v^{k}-\bvec{x}^\star}_2$ will be identical to that of their deterministic counterparts, that is either the inequality~\eqref{eq:difight-one-node-main-inequality} or~\eqref{eq:modifight-one-node-main-inequality}. Therefore, for any $v\in G_k$, one obtains, \begin{align}
\label{eq:randomized-one-participating-node-inequality}
h_v^{k+1} & \le \alpha \sum_{j=1}^L a_{jv}\left(\omega_jh_j^k+\mu_jb_j\right),
\end{align}
where $\alpha=\sqrt{3}$ or $\alpha=3$, depending on whether DiFIGHT or MoDiFIGHT is used. However, the nodes not in $G$ do not update their estimate, so that for $v\notin G_k$, one has \begin{align}
\label{eq:randomized-one-non-participating-node-inequality}
h_{v}^{k+1} & = h_v^k.
\end{align} Taking the inequalities~\eqref{eq:randomized-one-participating-node-inequality} and~\eqref{eq:randomized-one-non-participating-node-inequality} together, the following evolution inequality is obtained for the vector $\bvec{h}^k$:  \begin{align}
\label{eq:error-evolution-randomized-difight-modifight}
\bvec{h}^{k+1}\preccurlyeq \bvec{B}_k\bvec{h}^k+\bvec{c}_k,
\end{align}
%\item For randomized MoDiFIGHT: \begin{align}
%\label{eq:error-evolution-randomized-modifight}
%\bvec{h}^{n+1}\bvec{\le} \bvec{D}_1\bvec{A}_n^t\left(\bvec{\Omega}\bvec{h}^n+\bvec{M b}_1\right)+\bvec{D}_2\bvec{M b}_2
%\end{align}
%\end{enumerate} 
%where $\alpha = 2$ for DiFIGHT, and $\alpha = 4$ for MoDiFIGHT.
where the vector $\bvec{c}_k$ and the matrix $\bvec{A}_k$, where $\bvec{B}_k=\alpha \bvec{A}_k^t\Omega$ with $\bvec{A}_k=[\bvec{a}_{k,1}\cdots \bvec{a}_{k,L}]$, are determined as below:
\begin{align}
\label{eq:b_k-description}
c_{k,v} & = \left\{
\begin{array}{ll}
\alpha \bvec{a}_v^t\bvec{Mb}, & v\in G\\
0, & v\notin G
\end{array}
\right.\\
\label{eq:A_k-description}
%\begin{aligned}
\bvec{a}_{k,v} & = \left\{\begin{array}{ll}
\bvec{a}_v, & v\in G\\
\frac{\bvec{e}_v}{\alpha \omega_v}, & v\notin G
\end{array}
\right.
%\end{aligned}
\end{align}
where $\bvec{e}_v$ is the the column vector with all entries set to $0$ except for the $v^\mathrm{th}$ entry which is set to $1$.
%\begin{align*}
%\bvec{A}_n = \begin{bmatrix}
%1 & 0 & \cdots & a_{1v} & 0 & \cdots & 0\\
%0 & 1 & \cdots & a_{2v} & 0 & \cdots & 0\\
%\vdots & \vdots & \cdots & \vdots & \vdots & \cdots & \vdots\\
%0 & 0 & \cdots & a_{Lv} & 0 & \cdots & 1\\
%\end{bmatrix},
%\end{align*}
%
%Since the sequence $\bvec{h}^k$ is a random sequence (specifically a Markov Chain), we analyze the evolution of $\expect{\bvec{h}^k}$, where the expectation is taken with respect to the randomness associated to choosing nodes participating in the diffusion process. 
%
%Before further proceeding, we represent the inequality~\eqref{eq:error-evolution-randomized-difight-modifight} in the following compact form: \begin{align}
%\label{eq:error-evolution-compact-representation}
%\bvec{w}^{k+1}\preccurlyeq\bvec{B}_k\bvec{w}^k+\bvec{c}_k
%\end{align}
%where $\bvec{w}^k = \bvec{h}^k + \bvec{\Omega}^{-1}\bvec{M b}_k$, and $\bvec{B}_k = \alpha\bvec{A}_k^t\bvec{\Omega},\ \bvec{c}_k = \bvec{\Omega}^{-1}\bvec{M b}_k$.
%
%\begin{align}
%\label{eq:compact-representation-Bn}
%\bvec{B}_n & = \alpha\bvec{A}_n^t\bvec{\Omega},\\
%\label{eq:compact-representation-c}
%\bvec{c} & = \bvec{\Omega}^{-1}\bvec{M b}.
%\end{align}
We will now use the compact inequality~\eqref{eq:error-evolution-randomized-difight-modifight} to derive condition for stability of the mean of the sequence $\{\bvec{h}^k\}$.

Taking expectation of both sides of the inequality~\eqref{eq:error-evolution-randomized-difight-modifight} we find \begin{align}
\label{eq:error-evoltuion-mean-sequence-compact-form}
\expect{\bvec{h}^{k+1}} & \preccurlyeq\bvec{B}\expect{\bvec{h}^k} + \bvec{c}\\
\label{eq:error-evolution-mean-sequence-compact-form-convergence-representation}
\implies \expect{\bvec{h}^{k+1}} & \preccurlyeq\bvec{B}^{k+1}\expect{\bvec{h}^0} + \sum_{j=0}^k \bvec{B}^j \bvec{c}
\end{align}
where $\bvec{B}=\expect{\bvec{B}_k}$ and $\bvec{c}=\expect{\bvec{c}_k},\ k\ge 0$. It follows that the right hand side of the inequality~\eqref{eq:error-evolution-mean-sequence-compact-form-convergence-representation} converges if the matrix $\bvec{B}$ is stable. The matrix $\bvec{B}$ and vector $\bvec{c}$ have different forms for different strategies and for the different algorithms. We find them as below: 

Let us first find $\expect{\bvec{b}_{k,v}}$, where $\bvec{b}_{k,v}$ is the $v^\mathrm{th}$ column of the matrix $\bvec{B}_k$. Since $\expect{\bvec{B}_k} = \alpha \expect{\bvec{A}_k^t}\bvec{\Omega}$, we only require to find the expected value of $\bvec{a}_{k,v}$. Note that the column $\bvec{a}_{k,v}$ can take only two vector values, $\bvec{a}_v$ and $\bvec{e}_v/(\alpha\omega_v)$, depending on whether the node $v$ participates or not in the diffusion process at the $k^{\mathrm{th}}$ time step. Therefore, \begin{align}
\expect{\bvec{a}_{k,v}} & = \pi_v\bvec{a}_{v}+(1-\pi_v)\frac{\bvec{e}_v}{\alpha \omega_v},\ v=1,\cdots,L\nonumber\\
\implies \expect{\bvec{A}_k} & = \bvec{AP} + \frac{(\bvec{I} - \bvec{P})\bvec{\Omega}^{-1}}{\alpha}\nonumber\\
\label{eq:expectation-B-expression}
\implies \bvec{B} & =\bvec{I} - \bvec{P} + \alpha \bvec{PA}^t\bvec{\Omega}
\end{align} where $\bvec{P}=\diag{\pi_1,\cdots\, \pi_L}$.
Note that the diagonal matrix $\bvec{P}$ varies with different strategy and have diagonal entries $\pi_v$ that can be found from Table~\ref{tab:communication-complexities}.

In a similar manner one can find: \begin{align}
\label{eq:expectation-c-expression}
\expect{c_{k,v}} = \alpha \bvec{a}_v^t\bvec{Mb}\pi_v &
\implies \expect{\bvec{c}} = \alpha \bvec{P}\bvec{A}^t\bvec{Mb}.
\end{align}
%Where $\pi_v$ is the probability that, at a given time step, the node $v$ falls in the group of participating nodes. Thus $\pi_v$ depends on the randomization strategy used. We find below the values of $\pi_v$ for the different strategies adopted. 
%

Now, observe that as the network is connected, the matrix $\bvec{A}$ is irreducible. Also, $\bvec{P}$ and $\bvec{\Omega,\ }$ are non-negative diagonal matrices. Thus, using Lemma~\ref{lem:irreducibility}, and using the Perron-Frobenius theory, we can conclude that the matrix $\bvec{B}$ can be ensured to be Schur stable if the functions $\{f_i\}_{1\le i\le L}$ are chosen such that \begin{align}
\max_{i}\left(1-\pi_i +\alpha \pi_i\sum_{j=1}^L a_{ji}\omega_j\right) <1 & 
\Leftrightarrow \max_{i} \sum_{j=1}^L a_{ji}\omega_j <\frac{1}{\alpha}.
\end{align}  Note that, a weaker requirement is to choose the functions $\{f_i\}_{1=1}^L$ such that $(\max_i\omega_i)<1/\alpha$, as that implies $\sum_{j=1}^L a_{ji}\omega_j\le (\max_j\omega_j)\sum_{j=1}^L a_{ji}<1/\alpha$. Then, using Lemma~\ref{lem:convergence-linear-inequality}, from Equation~\eqref{eq:error-evolution-mean-sequence-compact-form-convergence-representation}, we have  \begin{align}
\bvec{h}\preccurlyeq(\bvec{I}-\bvec{B})^{-1} \bvec{c},
\end{align} where $\bvec{h}$ is any limit point of the sequence $\{\expect{\bvec{h}^k}\}_{k\ge 0}$. Now, using the expression of $\bvec{B},\bvec{c}$ from Eqs.~\eqref{eq:expectation-B-expression} and~\eqref{eq:expectation-c-expression}, respectively, one obtains, for any limit point $\bvec{h}$ of the sequence $\{\bvec{h}^k\}_{k\ge 0}$, \begin{align}
\bvec{h}\preccurlyeq \alpha(\bvec{I}-\alpha \bvec{A}^t\bvec{\Omega})^{-1}\bvec{A}^t\bvec{Mb}.
\end{align}

Now, when $\bvec{x}^\star$ is a stationary point of the functions $f_i,\ i=1,\cdots, L$, we have $\nabla f_i(\bvec{x}^\star) = \bvec{0},\ \forall i=1,\cdots,L$. Then, as per the definition of $\bvec{b}$ at the beginning of~\ref{sec:notation-main-results}, $\bvec{b} = \bvec{0}$, and thus $\bvec{c} = \bvec{0}$. Then, the evolution inequality~\eqref{eq:error-evolution-mean-sequence-compact-form-convergence-representation} reduces to $\expect{\bvec{h}^{k}} \preccurlyeq \bvec{B}^k \expect{\bvec{h}^0},\ k\ge 0$. Now, let us choose any coordinate $i$, $i=1\cdots, L$. To prove the almost sure convergence of $\{h_i^k\}_{k\ge 0}$, we use the Lemma~\ref{lem:almost-sure-convergence-grimmett}. Choose any $\epsilon>0$. Then we need to show that $\sum_{k\ge 0}\mathbb{P}(h_i^k>\epsilon)<\infty$ for any $\epsilon>0$. Now, for any $i=1,\cdots,L$, $h^k_i = \inprod{\bvec{e}_i}{\bvec{h}^k}$, where $\bvec{e}_i = [0\ 0\cdots 1\cdots 0\ 0]^t$, with the $1$ at the $i^\mathrm{th}$ coordinate. Consequently, 
\begin{align}
\label{eq:h-coordinate-inequality-almost-sure-convergence}
\expect{h_i^k}\le \inprod{\bvec{e}_i}{\bvec{B}^k \expect{\bvec{h}^0}}.
\end{align} 
Then, using Markov's inequality, followed by~\eqref{eq:h-coordinate-inequality-almost-sure-convergence}, we obtain $\sum_{k\ge 0}\mathbb{P}(h_i^k>\epsilon)\le \sum_{k\ge 0} \frac{\expect{h_i^k}}{\epsilon}\le \frac{\sum_{k\ge 0}\inprod{\bvec{e}_i}{\bvec{B}^k \expect{\bvec{h}^0}}}{\epsilon}$. Now, as $\bvec{B}$ is stable, the Neumann series $\sum_{k\ge 0}\bvec{B}^k$ converges to $(\bvec{I} - \bvec{B})^{-1}$, which allows, by Fubini's theorem to write $\sum_{k\ge 0}\inprod{\bvec{e}_i}{\bvec{B}^k \expect{\bvec{h}^0}} = \inprod{\bvec{e}_i}{\sum_{k=0}^\infty \bvec{B}^k \bvec{h}^0} = \inprod{\bvec{e}_i}{(\bvec{I} - \bvec{B})^{-1} \bvec{h}^0}$. Hence, $\sum_{k\ge 0}\mathbb{P}(h_i^k>\epsilon)\le \frac{\inprod{\bvec{e}_i}{(\bvec{I} - \bvec{B})^{-1} \bvec{h}^0}}{\epsilon}\le \frac{\opnorm{(\bvec{I}-\bvec{B})^{-1}\bvec{h}^0}}{\epsilon}<\infty$, for any $\epsilon>0$. Thus, by Lemma~\ref{lem:almost-sure-convergence-grimmett}, $h_i^k \to 0\ (a.s.)\implies \opnorm{\bvec{x}_i^k - \bvec{x}^\star}_2\to 0\ (a.s.)\implies \bvec{x}_i^k\to\bvec{x}^\star$ a.s.
\bibliography{diffusion-iht}
\end{document}